\documentclass[runningheads]{llncs}
\usepackage[T1]{fontenc}
\usepackage{graphicx}

\usepackage{amssymb}
\usepackage{centernot}
\usepackage{todonotes}
\usepackage{xspace}

\usetikzlibrary{arrows.meta}
\usetikzlibrary{automata, positioning}
\usetikzlibrary{calc}

\usepackage{amsmath}
\allowdisplaybreaks
\usepackage{comment}
\usepackage{cleveref}
\usepackage{thmtools}

\newcommand{\David}[1]{\todo[color=red!30]{\fontsize{6}{6}\selectfont DC: #1}}

\newcommand{\Calixte}[1]{\todo[color=yellow!30]{\fontsize{6}{6}\selectfont CG: #1}}
\newcommand{\Quentin}[1]{\todo[color=blue!30]{\fontsize{6}{6}\selectfont QM: #1}}

\newcommand{\FirstItem}{(i)\xspace}
\newcommand{\SecondItem}{(ii)\xspace}
\newcommand{\ThirdItem}{(iii)\xspace}

\newcommand{\FormatEntitySet}[1]{\ensuremath{\mathsf{#1}}}

\newcommand{\FormatFormulaSet}[1]{\ensuremath{\mathcal{#1}}\xspace}

\newcommand{\FormatFunction}[1]{\ensuremath{\mathbf{#1}}\xspace}
\newcommand{\FF}[1]{\FormatFunction{#1}}
\newcommand{\FormatPredicate}[1]{\ensuremath{\mathtt{#1}}\xspace}
\newcommand{\FP}[1]{\FormatPredicate{#1}}

\newcommand{\Classes}{\FormatEntitySet{Classes}\xspace}
\newcommand{\Roles}{\FormatEntitySet{Roles}\xspace}
\newcommand{\InvRoles}{\FormatEntitySet{InvRoles}\xspace}
\newcommand{\RoleExpressions}{\FormatEntitySet{RoleExps}\xspace}

\newcommand{\Constants}{\FormatEntitySet{Cons}\xspace}
\newcommand{\Variables}{\FormatEntitySet{Vars}\xspace}
\newcommand{\Terms}{\FormatEntitySet{Terms}\xspace}

\newcommand{\Body}{\ensuremath{\beta}\xspace}
\newcommand{\Head}{\ensuremath{\eta}\xspace}
\newcommand{\Formula}{\ensuremath{\psi}\xspace}
\newcommand{\FormulaAux}{\ensuremath{\psi'}\xspace}

\newcommand{\Fact}{\ensuremath{\alpha}\xspace}
\newcommand{\Rule}{\ensuremath{\varrho}\xspace}
\newcommand{\Query}{\ensuremath{\gamma}\xspace}
\newcommand{\QueryAux}{\ensuremath{\gamma'}\xspace}

\newcommand{\PositiveConjunctiveTwoWayRegularPathQuery}{\ensuremath{\varphi}\xspace}
\newcommand{\PCTRPQ}{\PositiveConjunctiveTwoWayRegularPathQuery}
\newcommand{\PQ}{\PCTRPQ}
\newcommand{\Substitution}{\ensuremath{\sigma}\xspace}
\newcommand{\Subs}{\Substitution}
\newcommand{\Trigger}{\ensuremath{\tau}\xspace}

\newcommand{\RuleSet}{\FormatFormulaSet{R}}
\newcommand{\RS}{\RuleSet}
\newcommand{\TBox}{\FormatFormulaSet{T}}
\newcommand{\FactSet}{\FormatFormulaSet{F}}
\newcommand{\FS}{\FactSet}

\newcommand{\SentenceSet}{\FormatFormulaSet{S}}
\renewcommand{\SS}{\SentenceSet}

\newcommand{\KB}{\FormatFormulaSet{K}}

\newcommand{\HornSHIQ}{\ensuremath{\textit{Horn-}\mathcal{SHIQ}}\xspace}
\newcommand{\HornALC}{\ensuremath{\textit{Horn-}\mathcal{ALC}}\xspace}
\newcommand{\HornALCDat}{\ensuremath{\textit{Horn-}\mathcal{ALC}_{\centernot\exists}}\xspace}
\newcommand{\HornALCHI}{\ensuremath{\textit{Horn-}\mathcal{ALCHI}}\xspace}
\newcommand{\EL}{\ensuremath{\mathcal{EL}}\xspace}

\newcommand{\quasilinELHI}{\ensuremath{\mathcal{ELH}^{ql}_i}}
\newcommand{\ELI}{\ensuremath{\mathcal{ELI}}\xspace}
\newcommand{\ELHI}{\ensuremath{\mathcal{ELHI}}\xspace}
\newcommand{\Imc}{\mathcal{I}}
\newcommand{\Jmc}{\mathcal{J}}
\newcommand{\propconcepts}{$(\star)$}
\newcommand{\roleclosure}{\sqsubseteq^*}

\newcommand{\DLAutomaton}{\ensuremath{\mathfrak{A}}\xspace}
\newcommand{\DLAut}{\DLAutomaton}
\newcommand{\NonDeterministicFiniteAutomaton}{\ensuremath{\mathfrak{N}}\xspace}
\newcommand{\NFA}{\NonDeterministicFiniteAutomaton}
\newcommand{\EmptyWord}{\ensuremath{\varepsilon}\xspace}
\newcommand{\RegularExpression}{\ensuremath{\mathfrak{E}}\xspace}
\newcommand{\RE}{\RegularExpression}
\newcommand{\InitialState}{\ensuremath{q_{\textit{init}}}\xspace}
\newcommand{\AcceptingState}{\ensuremath{q_{\textit{acc}}}\xspace}
\newcommand{\EmptyTransition}{\ensuremath{\varepsilon}\xspace}
\newcommand{\TransitionFunction}{\ensuremath{\delta}\xspace}
\newcommand{\TF}{\ensuremath{\TransitionFunction}\xspace}

\newcommand{\Transition}[3]{\ensuremath{#1 \rightsquigarrow_{#2} #3}\xspace}
\newcommand{\MultiaryTrans}[2]{\ensuremath{#1 \rightsquigarrow_{\EmptyTransition} \{ #2 \}}\xspace}
\newcommand{\BinaryTransitionRelation}[3]{\ensuremath{#2 \to_{#1}^2 #3}\xspace}
\newcommand{\BinTransRel}[2]{\BinaryTransitionRelation{#1}{#2}}
\newcommand{\MultiTransitionRelation}[3]{\ensuremath{#2 \to_{#1}^{m} #3}\xspace}
\newcommand{\MultiTransRel}[2]{\MultiTransitionRelation{#1}{#2}}
\newcommand{\BothTransitionRelation}[3]{\ensuremath{#2 \to_{#1} #3}\xspace}
\newcommand{\BothTransRel}[2]{\BothTransitionRelation{#1}{#2}}
\newcommand{\TransitiveReflexiveTransitionRelation}[3]{\ensuremath{#2 \to_{#1}^* #3}\xspace}
\newcommand{\TRTransRel}[2]{\TransitiveReflexiveTransitionRelation{#1}{#2}}
\newcommand{\InHigherStrata}[3]{\ensuremath{#2 \succ_{#1} #3}\xspace}

\newcommand{\QueryMapping}{\FF{QM}}
\newcommand{\Language}{\FF{Lang}}
\newcommand{\Paths}[1]{\ensuremath{\FF{Paths}_{#1}}\xspace}
\newcommand{\Invert}{\FF{Inv}}
\newcommand{\Domain}{\FF{Domain}}
\newcommand{\HornALCToDLAutomaton}{\FF{DLA}}
\newcommand{\ToDLA}{\HornALCToDLAutomaton}
\newcommand{\StratificationFunction}[1]{\ensuremath{\FF{Str}_{#1}}\xspace}
\newcommand{\StrFun}[1]{\StratificationFunction{#1}}
\newcommand{\RegularExpressionNFAStateToState}[3]{\ensuremath{\RE_{#1 \triangleright #2}^{#3}}\xspace}

\newcommand{\NFAToPostiveQuery}{\FF{PQ}_\textit{NFA}}
\newcommand{\StratDLAToPostiveQuery}{\FF{PQ}}
\newcommand{\StratToPQ}{\StratDLAToPostiveQuery}

\begin{document}

\title{A General Sufficient Condition for Rewriting \HornALCHI Atomic Queries into GQL}
\titlerunning{A Sufficient Condition for Rewriting \HornALCHI Queries into GQL}
\author{David Carral\inst{1}, Calixte Gruson\inst{2}, and Quentin Mani\`ere\inst{1}}
\institute{$^1$LIRMM, Inria, University of Montpellier, CNRS, Montpellier, France \\
$^2$DI ENS, ENS, CNRS, PSL University \& Inria, France}
%
\authorrunning{David Carral, Calixte Gruson, and Quentin Mani\`ere}

\maketitle

\begin{abstract}
The emergence of the ISO standard GQL introduces a powerful query language extending first-order logic with controlled recursion, raising the question of its applicability to evaluation of ontology-mediated queries (OMQs). 
We focus on OMQs consisting of atomic queries over ontologies expressed in \HornALCHI, an expressive Description Logic that is not, in general, first-order rewritable. 
To address this, we introduce DL automata, a novel formalism that captures the semantics of such OMQs via runs over fact sets.
We then identify a large class of DL automata that can be rewritten into unions of conjunctive two-way regular path queries (UC2RPQs), a central fragment of GQL. 
Our class of automata relies on a stratification of their states, ruling out specific forms of cyclic dependencies known to raise the complexity. 
This yields a broad class of \HornALCHI OMQs that are GQL-rewritable.
\end{abstract}

\section{Introduction}
\label{section:introduction}

In ontology-mediated query answering, user queries are enriched with conceptual knowledge about the domain of interest, described in the form of an ontology.
This conceptual layer enables the use of a more user-friendly vocabulary and supports the retrieval of implicit information derived from the explicit knowledge in the ontology, thereby enriching query answers.
Description Logics (DLs) are a predominant formalism for representing ontologies \cite{DBLP:conf/dlog/2003handbook}.
They provide the logical foundations of the OWL web ontology language, a W3C-standardized language for the Semantic Web \cite{DBLP:books/crc/Hitzler2010}.
In this work, we focus on the DL $\HornALCHI$, an expressive language that allows the use of inverse roles and role constructors; the following example illustrates its usefulness.

\begin{example}
\label{example:introduction-ontology}
Consider the set \TBox of \HornALCHI axioms\footnote{In the literature, DL axioms are often presented using a dedicated syntax \cite{dl-primer}; we use equivalent first-order formulas, which are more familiar to a broader audience.} given below, which formalises the notion of a \FP{Trust}{}ed user in a computer network.
\begin{align*}
\forall x, y . &\FP{Link}(x, y) \wedge \FP{Sens}(y) \to \FP{Sens}(x) & \forall x . &\FP{Gate}(x) \to \exists y . \FP{Link}(x, y) \wedge \FP{Onl}(y) \\
\forall x, y . &\FP{Link}(x, y) \wedge \FP{Onl}(y) \to \FP{Onl}(x) & \forall x, y . &\FP{Crit}(x) \wedge \FP{Acc}(x, y) \to \FP{Trust}(y) \\
\forall x . &\FP{Sens}(x) \wedge \FP{Onl}(x) \to \FP{Crit}(x)
\end{align*}
Furthermore, consider the fact set $\FS$, defined below and illustrated in \Cref{figure:example}.
\begin{align*}
\{& \FP{Acc}(c_1, \textit{alice}), \FP{Link}(c_1, c_2), \FP{Sens}(c_3), \FP{Link}(c_1, c_4), \FP{Link}(c_2, c_3), \notag \\
&\FP{Gate}(c_4), \FP{Acc}(c_5, \textit{bob}), \FP{Link}(c_5, c_4), \FP{Link}(c_5, c_6), \FP{Link}(c_6, c_7), \FP{Onl}(c_7)\} \notag
\end{align*}

The \HornALCHI ontology $\langle \TBox, \FS \rangle$ entails the fact $\FP{Crit}(c_1)$: indeed, node $c_1$ has a $\FP{Link}$ path to a $\FP{Sens}${}itive node and another such path to a $\FP{Gate}${}way node (see \Cref{figure:example}); moreover, by the second axiom of \TBox, every $\FP{Gate}${}way node has a $\FP{Link}$ to some $\FP{Onl}${}ine node.
Therefore, the ontology also entails $\FP{Trust}(\textit{alice})$.
\end{example}

\newcommand{\XSep}{2.5}
\newcommand{\YSep}{0.8}
\definecolor{NiceBlue}{RGB}{70,130,180} 

\begin{figure}[t]
\centering
\begin{tikzpicture}
\node (alice) [circle, fill, inner sep=1.5pt, label={left:$\color{NiceBlue} \FP{Trust} \color{black}: \textit{alice}$}] at (1*\XSep, 3*\YSep) {};
\node (a2) [circle, fill, inner sep=1.5pt, label={[xshift=12pt]above:$c_1 : \color{NiceBlue} \FP{Crit}$}] at (2*\XSep, 3*\YSep) {};
\node (a3) [circle, fill, inner sep=1.5pt, label={above:$c_2$}] at (3*\XSep, 3*\YSep) {};
\node (a4) [circle, fill, inner sep=1.5pt, label={right:$c_3 : \FP{Sens}$}] at (4*\XSep, 3*\YSep) {};

\node (m3) [circle, fill, inner sep=1.5pt, label={right:$c_4 : \FP{Gate}$}] at (3*\XSep, 2*\YSep) {};

\node (bob) [circle, fill, inner sep=1.5pt, label={left:$\textit{bob}$}] at (1*\XSep, 1*\YSep) {};
\node (b2) [circle, fill, inner sep=1.5pt, label={above:$c_5$}] at (2*\XSep, 1*\YSep) {};
\node (b3) [circle, fill, inner sep=1.5pt, label={above:$c_6$}] at (3*\XSep, 1*\YSep) {};
\node (b4) [circle, fill, inner sep=1.5pt, label={right:$c_7 : \FP{Onl}$}] at (4*\XSep, 1*\YSep) {};

\draw[solid, thick, -{Stealth[scale=1.2]}] (a2) -- node[anchor=center, fill=white] {\FP{Acc}} (alice);
\draw[solid, thick, -{Stealth[scale=1.2]}] (a2) -- node[anchor = center, fill=white] {\FP{Link}} (a3);
\draw[solid, thick, -{Stealth[scale=1.2]}] (a2) -- node[anchor = center, fill=white] {\FP{Link}} (m3);
\draw[solid, thick, -{Stealth[scale=1.2]}] (a3) -- node[anchor = center, fill=white] {\FP{Link}} (a4);

\draw[solid, thick, -{Stealth[scale=1.2]}] (b2) -- node[anchor = center, fill=white] {\FP{Acc}} (bob);
\draw[solid, thick, -{Stealth[scale=1.2]}] (b2) -- node[anchor = center, fill=white] {\FP{Link}} (b3);
\draw[solid, thick, -{Stealth[scale=1.2]}] (b3) -- node[anchor = center, fill=white] {\FP{Link}} (b4);
\draw[solid, thick, -{Stealth[scale=1.2]}] (b2) -- node[anchor = center, fill=white] {\FP{Link}} (m3);

\end{tikzpicture}
\caption{The fact set \FS (depicted in black), together with some of the unary facts entailed by the ontology $\langle \TBox, \FS \rangle$ (depicted in blue), defined in \Cref{example:introduction-ontology}.}
\label{figure:example}
\end{figure}
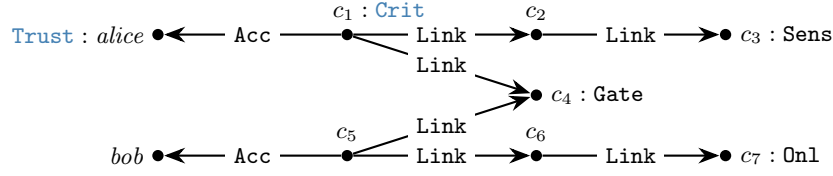

In the presence of DL axioms, a central reasoning task is to decide whether a given unary fact is entailed; together with the axioms, such a query forms a \emph{rule query}.
A common approach to evaluating a rule query over a database is to rewrite it into a classical query language, typically an FO query, while preserving equivalence, that is, ensuring that the original rule query and its rewriting are entailed by the same fact sets.
This approach requires a target query language that is expressive enough to capture the ontological knowledge encoded in the axioms, but it also allows one to reuse existing database technologies, thereby avoiding the need to develop new evaluation tools.

The recently released ISO standard GQL,\footnote{\url{https://www.iso.org/standard/76120.html}} for graph query languages, extends the expressive power of traditional FO queries (as captured by the SQL standard) by allowing controlled forms of regular recursion.
GQL can express conjunctions of two-way regular path queries and their unions (UC2RPQs), and the evaluation problem for queries at its core has been shown to be NL-complete~\cite{sigmod22:gql,icdt23:digest-of-gql}.
As a new standard, GQL is expected to play a central role in a wide range of applications, making it important to understand which rule queries can be rewritten into this language.
Example~\ref{example:introduction-query} presents an equivalent rewriting into UC2RPQs, and hence into GQL, of a rule query based on Example~\ref{example:introduction-ontology}.

\begin{example}
\label{example:introduction-query}
Consider the \HornALCHI query $\langle \TBox, \FP{Trust}(\textit{alice}) \rangle$, where $\TBox$ is the set of DL axioms described in \Cref{example:introduction-ontology}.
This query is equivalent to the following positive conjunctive regular path query:
\begin{align*}
&q_\textit{alice}^\FP{TU} = \FP{Trust}(\textit{alice}) \vee \exists y . \Big[ \FP{Acc}(y, \textit{alice}) \wedge \FP{Crit}(y)\Big] ~\vee \\
&~~\exists y, z, w . \Big[ \FP{Acc}(y, \textit{alice}) \wedge \FP{Link}^*(y, z) \wedge \FP{Sens}(z) 
\wedge \FP{Link}^*(y, w) \wedge \big(\FP{Onl}(w) \vee \FP{Gate}(w)\big)\Big]
\end{align*}
That is, these two boolean queries induce the same partition over the set of all fact sets.
Put differently, for any given fact set $\FS'$, the DL ontology $\langle \TBox, \FS' \rangle$ entails the fact $\FP{Trust}(\textit{alice})$ if and only if $\FS'$ satisfies the query $q_\textit{alice}^\FP{TU}$.
Consequently, instead of verifying whether $\langle \TBox, \FS' \rangle$ entails $\FP{Trust}(\textit{alice})$ using a dedicated DL reasoner, we may simply check whether $\FS'$ satisfies $q_\textit{alice}^\FP{TU}$ using a database management system.
To conclude the example, observe that the fact set $\FS$ from \Cref{example:introduction-ontology} satisfies $q_\textit{alice}^\FP{TU}$, as expected, since $\langle \TBox, \FS \rangle$ does entail $\FP{Trust}(\textit{alice})$.
Moreover, the rule query $\langle \TBox, \FP{Trust}(\textit{alice}) \rangle$ is not expressible as an FO query, thereby motivating the need for rewriting techniques that go beyond this limited database query language.
\end{example}

\paragraph{Contributions.}
We focus on rule queries with unary facts such as the one in the previous example, with rules expressed in the positive fragment of $\HornALCHI$ DL, defined in \Cref{section:preliminaries}.
We show that these rules can be normalised into Datalog ones, in \Cref{section:horn-alc-to-normalised-horn-alc}.
We then identify a large class of rule queries that can be rewritten into GQL; we proceed in two main steps.
First, in \Cref{section:normalised-horn-alc-to-automata}, we turn our rule queries into equivalent automata of a new form, closely related to the usual tree automata and which we refer to as DL automata.
Second, in \Cref{section:automata-to-acyclic-pqs}, we identify a class of DL automata that can be rewritten into positive two-way regular path queries (PQs), and thus into UC2RPQs which in turn fall within GQL (\Cref{section:ayclicic-c2rpqs-to-gql}).
Our criteria for rewritability relies on the existence of a stratification of the DL automata, which forbids some (but not all!) cyclic dependencies within its transition function.
Interestingly, rule queries based upon the DL-Lite family trivially enjoy stratified automata, and thus the class we pinpoint extends DL-Lite. 
Additionally, our class guarantees the rewritability of rule queries that were not covered by prior attempts at rewriting into GQL \cite{eswc25-rew-quasilinear}, such as the one from \Cref{example:introduction-query}.
\Cref{section:related-work} elaborates further on related work; we discuss our conclusions and future work in \Cref{section:work-conclusions-future}.
This extended version of a paper to appear at ISWC 2026 contains proofs for all results in the appendix.

\section{Preliminaries}
\label{section:preliminaries}

We assume that the reader is familiar with the syntax and semantics of first-order~(FO) logic, as well as with basic notions from automata theory.

\paragraph{Existential Rules and FO Facts.}
We fix pairwise disjoint, countably infinite sets \Classes, \Roles, \Constants, and \Variables of unary predicates, binary predicates, constants, and variables respectively.
We often refer to unary and binary predicates as \emph{classes} and \emph{roles}, respectively.
An \emph{inverse role} is an element of $\InvRoles = \{\FP{R}^- \mid \FP{R} \in \Roles\}$.
A \emph{role expression} is an element of $\RoleExpressions = \Roles \cup \InvRoles$.
A \emph{term} is an element of $\Terms = \Constants \cup \Variables$.
Lists of terms $t_1, \ldots, t_n$ are written as $\vec{t}$ and are often treated as sets.

A \emph{class atom} is a formula of the form $\FP{C}(t)$ with $\FP{C} \in \Classes$ and $t \in \Terms$.
A \emph{role atom} is a formula of the form $\FP{R}(t, u)$ with $\FP{R} \in \Roles$ and $t, u \in \Terms$.
A \emph{FO atom} is either a class atom or a role atom.
\emph{Class facts}, \emph{role facts}, and \emph{FO facts} are variable-free class atoms, role atoms, and FO atoms, respectively.

For an FO formula \Formula and a list $\vec{x}$ of variables, we write $\Formula[\vec{x}]$ to indicate that~$\vec{x}$ is the set of all free variables in \Formula.
An \emph{(existential) rule} \Rule is a FO sentence of the form $\forall \vec{x}, \vec{y} . \big(\Body[\vec{x}] \rightarrow \exists \vec{z} . \Head[\vec{y}, \vec{z}]\big)$ where $\vec{x}$, $\vec{y}$, and~$\vec{z}$ are lists of variables such that $(\vec{x} \cup \vec{y}) \cap \vec{z} = \emptyset$, and \Body and \Head are conjunctions of FO atoms with $\Head \neq \emptyset$.

The rule \Rule is \emph{Datalog} if $\vec{z}$ is the empty list.
We refer to \Body and $\exists \vec{z} . \Head$ as the \emph{body} and the \emph{head} of \Rule, respectively.
When discussing rules, we omit universal quantifiers and often identify conjunctions of atoms such as \Body and~\Head above with the corresponding set.
We allow rules to contain \emph{unsafe variables}--universally quantified variables that appear in the head but not in the body of the rule.
A \emph{knowledge base} (KB) is a pair consisting of a finite rule set and a finite fact set.

\paragraph{Description Logics.}
A \emph{(positive) \HornALCHI rule} is one of the following rules:
\begin{align}
\FP{A}_1(x) \wedge \ldots \wedge \FP{A}_n(x) &\to \FP{B}(x) & \FP{A}_1 \sqcap \ldots \sqcap \FP{A}_n  &\sqsubseteq \FP{B} \label{rule:conjunction} \\
\FP{R}(x, y) \wedge \FP{A}(y) &\to \FP{B}(x) & \exists \FP{R}.\FP{A} &\sqsubseteq \FP{B} \label{rule:lhs-existential} \\
\FP{A}(x) \wedge \FP{R}(x, y) &\to \FP{B}(y) & \FP{A} &\sqsubseteq \forall \FP{R}.\FP{B} \label{rule:universal} \\
&\to \FP{B}(x) & \top &\sqsubseteq \FP{B} \label{rule:top} \\
\FP{A}(x)&\to \exists y . \FP{R}(x, y) \wedge \FP{B}(y) & \FP{A} &\sqsubseteq \exists \FP{R}.\FP{B} \label{rule:rhs-existential} \\
\FP{R}(x, y) &\to \FP{S}(x, y) & \FP{R} &\sqsubseteq \FP{S} \label{rule:subrole} \\
\FP{R}(y, x)&\to \FP{S}(x, y) & \FP{R}^- &\sqsubseteq \FP{S} \label{rule:inverse}
\end{align}
Here, $\FP{A}$, $\FP{B}$, and $\FP{C}$ are classes, and $\FP{R}$ and $\FP{S}$ are roles.
We restrict ourselves to the positive fragment of \HornALCHI, in which the bottom class~$\bot$ is not allowed; in \Cref{section:work-conclusions-future}, we discuss how our rewriting techniques could be adapted to handle this class.
We consider \HornALCHI rules in a restricted normal form, which is without loss of generality \cite{normal-form-horn-alc}.
The right-hand side shows the corresponding Description Logic syntax, included for reference only and not used further in this paper.
We adopt the rule-based presentation because rules are FO logic formulas and are therefore accessible to a wider audience, and because this syntax allows us to define fact entailment for \HornALCHI naturally via the well-known chase procedure, which we briefly present below.
We consider entailment under active domain semantics as discussed later, and hence rules of type~\eqref{rule:top} are not domain-dependent.
A \emph{\HornALCHI KB} is a KB whose rule set is \HornALCHI.

\paragraph{Query Languages.}
In this paper, a \emph{rule query} is a pair consisting of a finite rule set and a class fact.
A \emph{Datalog query} is a rule query whose rule set is Datalog.
A \emph{\HornALCHI query} is a rule query whose rule set is \HornALCHI.

A \emph{regular expression} is an expression generated by the grammar $$\RE ::= \EmptyTransition \mid \FP{R} \mid \FP{R}^- \mid \big(\RE \cup \RE\big) \mid \big(\RE \cdot \RE\big) \mid \big(\RE\big)^*,$$ where $\FP{R} \in \Roles$.
For a regular expression~$\RE$, let $\Language(\RE)$ denote the set of all words over role expressions generated by~$\RE$.
A \emph{RE atom} is a formula of the form~$\RE(t, u)$ with \RE a regular expression and $t, u \in \Terms$.
A \emph{RE fact} is a variable-free RE atom.
A \emph{positive two-way regular path query} (PQ) is a formula without occurrences of free variables generated by the grammar
$$\PCTRPQ ::= \top \mid \FP{A}(t) \mid \RE(t, u) \mid \big(\PCTRPQ \wedge \PCTRPQ \big) \mid \big(\PCTRPQ \vee \PCTRPQ \big) \mid \exists y . \PCTRPQ,$$
where $\FP{A}(t)$ is a class atom, $\RE(t, u)$ is a RE atom, and $y$ is a variable.
Note that, under our definition, all PQs are Boolean queries since every occurrence of a variable in a PQ lies within the scope of an existential quantifier.
Without loss of generality, we assume that every variable is quantified exactly once in a PQ.
 
\paragraph{Semantics.}
Consider some FO sentence \Formula.
Given another FO sentence \FormulaAux, we write $\Formula \models \FormulaAux$ to denote entailment under under active domain semantics, that is, standard FO semantics with quantifiers ranging over the active domain.
Given a finite set $\SS$ of FO sentences, we write~$\SS \models \Formula$ to denote $\bigwedge_{\FormulaAux \in \SS} \FormulaAux \models \Formula$.
Given a KB $\langle \RS, \FS \rangle$, we write $\langle \RS, \FS \rangle \models \Formula$ to denote $\RS \cup \FS \models \Formula$.

A \emph{substitution} is a function mapping variables to terms.
For a substitution \Subs and a FO formula $\Formula[\vec{x}]$, let $\Subs(\Formula)$ denote the formula obtained from \Formula by replacing all occurrences of every variable $x \in \Domain(\Subs) \cap \vec{x}$ with $\Subs(x)$.

Consider a fact set \FS.
We define $\Invert(\FS) = \{\FP{R}^-(c, d) \mid \FP{R}(d, c) \in \FS\}$.
A (possibly empty) word over role expressions $\FP{R_1}\ldots \FP{R_n}$ is a \emph{path} from a constant~$c$ to a constant $d$ in \FS if both $c$ and $d$ occur in \FS, and there exists a list of constants $e_0, \ldots, e_n$ such that $e_0 = c$, $e_n = d$, and $\FP{R_i}(e_{i-1}, e_i) \in \FS \cup \Invert(\FS)$ for every $1 \leq i \leq n$.
In particular, for $n = 0$, the \emph{empty word} \EmptyWord is a path from any constant occuring in \FS to itself.
For some constants $c$ and $d$ occurring \FS, we write $\Paths{\FS}(c, d)$ for the set of all paths from $c$ to $d$ in \FS.
Entailment of a PQ \PCTRPQ by \FS, written~$\FS \models \PCTRPQ$, is defined inductively as follows:
\begin{itemize}
\item If $\PCTRPQ = \top$, then $\FS \models \PCTRPQ$.
\item If \PCTRPQ is a class fact, then $\FS \models \PCTRPQ$ if $\PCTRPQ \in \FS$.
\item If \PCTRPQ is a RE fact $\RE(c, d)$, then $\FS \models \PCTRPQ$ if $\Paths{\FS}(c, 
d) \cap \Language(\RE) \neq \emptyset$.
\item If \PCTRPQ is of the form $\big(\PCTRPQ_1 \wedge \PCTRPQ_2\big)$, then $\FS \models \PCTRPQ$ if $\FS \models \PCTRPQ_1$ and $\FS \models \PCTRPQ_2$.
\item If \PCTRPQ is of the form $\big(\PCTRPQ_1 \vee \PCTRPQ_2\big)$, then $\FS \models \PCTRPQ$ if $\FS \models \PCTRPQ_1$ or $\FS \models \PCTRPQ_2$.
\item If \PCTRPQ is of the form $\exists x . \PCTRPQ'$, then $\FS \models \PCTRPQ$ if $\FS \models \Subs(\PCTRPQ')$ for some substitution~\Subs mapping $x$ to some constant in \FS.
\end{itemize}
Since PQs contain no free variables, defining $\FS \models \Fact$ for any ground atom \Fact suffices to fully specify the semantics of entailment.

For a rule query $\langle \RS, \Fact \rangle$, we write~$\QueryMapping(\langle \RS, \Fact \rangle)$ or $\QueryMapping(\RS, \Fact)$ to refer to the set of all finite fact sets $\FS$ such that $\langle \RS, \FS \rangle \models \Fact$.
For a PQ $\PCTRPQ$, let $\QueryMapping(\PCTRPQ)$ be the set of all finite fact sets that entail $\PCTRPQ$.
Here, \QueryMapping stands for ``query mapping'', since it associates every rule query or PQ with its semantic extension, namely the set of finite fact sets that satisfy it.
A rule query or PQ $\Query$ is \emph{contained} in another such query~$\QueryAux$ if $\QueryMapping(\Query) \subseteq \QueryMapping(\QueryAux)$; these are \emph{equivalent} if $\QueryMapping(\Query) = \QueryMapping(\QueryAux)$.

\paragraph{The Chase.}
The chase algorithm can be used to decide fact entailment and is defined here only for Datalog, since \HornALCHI rules with existentially quantified variables can be normalised away, as explained in the following section.
A \emph{trigger}~\Trigger is a pair $\langle \Rule, \Substitution \rangle$ consisting of a Datalog rule $\Rule = \Body \to \Head$ and a substitution \Subs that is defined exactly on the set of all variables occurring in~\Rule.
The trigger \Trigger is \emph{applicable} to a fact set \FS if $\Subs(\Body) \subseteq \FS$ and all constants in the range of $\Subs$ occur in \FS.
We define $\FF{Output}(\Trigger) = \Subs(\Head)$ and, for a fact set \FS, we define $\FF{App}_\Rule(\FS)$ as the set including $\FF{Output}(\Trigger)$ for every trigger \Trigger with \Rule that is applicable to \FS.
For a Datalog rule set $\RS$ and a fact set $\FS$, let $\FF{App}_\RS(\FS) = \bigcup_{\Rule \in \RS}\FF{App}_\Rule(\FS)$.
For a Datalog KB $\KB = \langle \RS, \FS \rangle$, let $\FF{Ch}_1(\KB), \FF{Ch}_2(\KB), \ldots$ be the sequence of fact sets such that \FirstItem $\FF{Ch}_1(\KB) = \FS$, and \SecondItem $\FF{Ch}_i(\KB) = \FF{Ch}_{i-1}(\KB) \cup \FF{App}_\RS(\FF{Ch}_{i-1}(\KB))$ for every $i \geq 2$.
We define the \emph{chase} of \KB, written~$\FF{Ch}(\KB)$, as the union of all fact sets in this sequence.
By our definition of applicable triggers, the chase is always finite, and it yields a set of facts that can be used to decide fact entailment under active domain semantics \cite{foundations-of-databases-95}.

\begin{proposition}
\label{propoosition:chase-correctness}
A Datalog KB \KB entails a fact \Fact if and only if $\Fact \in \FF{Ch}(\KB)$.
\end{proposition}

For a Datalog KB \KB, the \emph{chase depth} of a fact $\Fact \in \FF{Ch}(\KB)$ with respect to~\KB is the least $k \geq 1$ such that $\Fact \in \FF{Ch}_{k}(\KB)$.

\section{From \HornALCHI to \HornALCDat Queries}
\label{section:horn-alc-to-normalised-horn-alc}

This section briefly recalls how to normalize away axioms of the form \eqref{rule:rhs-existential}, \eqref{rule:subrole}, and \eqref{rule:inverse} from a \HornALCHI rule set.

\begin{definition}
A \emph{\HornALCDat rule} is a rule of the form \eqref{rule:conjunction}, \eqref{rule:lhs-existential}, \eqref{rule:universal}, or \eqref{rule:top}.
\HornALCDat queries and KBs are defined in the obvious way.
\end{definition}

To reduce to such rules, we proceed in two main steps.
Assume we have a \HornALCHI rule set $\RS$.
In the first step, we get rid of axioms of the forms \eqref{rule:subrole} and \eqref{rule:inverse}.
Let $\roleclosure$ be the reflexive and transitive closure of the relation $\sqsubseteq$ induced by $\RS$ on $\RoleExpressions$.
Let $\RS'$ be a copy of $\RS$ from which we remove all axioms of the forms \eqref{rule:subrole} and \eqref{rule:inverse} but add, for every $\FP{R}, \FP{S} \in \Roles$ such that $\FP{R} \roleclosure \FP{S}$ or $\FP{R}^{-} \roleclosure \FP{S}^{-}$, the axioms:
\begin{itemize}
    \item $\FP{A}(x) \wedge \FP{R}(x, y) \to \FP{B}(y)$ for every axiom $\FP{A}(x) \wedge \FP{S}(x, y) \to \FP{B}(y)$ in $\RS$; and
    \item $\FP{R}(x, y) \wedge \FP{A}(y) \to \FP{B}(x)$ for every axiom $\FP{S}(x, y) \wedge \FP{A}(y) \to \FP{B}(x)$ in $\RS$;
\end{itemize}
and, for every $\FP{R}, \FP{S} \in \Roles$ such that $\FP{R}^- \roleclosure \FP{S}$ or $\FP{R} \roleclosure \FP{S}^{-}$, the axioms:
\begin{itemize}
    \item $\FP{R}(x, y) \wedge \FP{A}(y) \to \FP{B}(x)$ for every axiom $\FP{A}(x) \wedge \FP{S}(x, y) \to \FP{B}(y)$ in $\RS$; and
    \item $\FP{A}(x) \wedge \FP{R}(x, y) \to \FP{B}(y)$ for every axiom $\FP{S}(x, y) \wedge \FP{A}(y) \to \FP{B}(x)$ in $\RS$.
\end{itemize}
It is then routine to verify that for every class fact $\Fact$, the queries $\langle \RS, \Fact \rangle$ and $\langle \RS', \Fact \rangle$ are equivalent.
We here strongly rely on the queries being based upon class facts, as this procedure would fail already for those based on role facts.

In the second step, we get rid of axioms of the form \eqref{rule:rhs-existential} by applying the translation to Datalog rules proposed in \cite[Tables~2 and 3]{normal-form-horn-shiq} for the more expressive \HornSHIQ.
This reduction to their procedure is possible as our $\RS'$ is already compliant with their notion of a \HornSHIQ rule set.
We then remark that, given that we only input axioms of the forms \eqref{rule:conjunction}--\eqref{rule:rhs-existential} (or (F1), (F2), (F3) in the reference), their translation only outputs axioms of the form \eqref{rule:conjunction}, \eqref{rule:lhs-existential}, \eqref{rule:universal}, or \eqref{rule:top}, as desired.
In summary, we obtain the following:

\begin{theorem}\label{theorem:normalization-is-correct}
There is a computable function mapping every \HornALCHI query to an equivalent \HornALCDat query.
\end{theorem}

\begin{example}
\label{example:normalise-away-existentials}
Consider the set of \HornALCDat rules $\TBox_\textit{NF}$ obtained from the set $\TBox$ from \Cref{example:introduction-ontology} by removing $\FP{Gate}(x) \to \exists y . \FP{Link}(x, y) \wedge \FP{Onl}(y)$ and adding $\FP{Gate}(x) \to \FP{Onl}(x)$.
Furthermore, observe that the DL queries $\langle \TBox, \FP{Trust}(\textit{alice}) \rangle$ and $\langle \TBox_\textit{NF}, \FP{Trust}(\textit{alice}) \rangle$ are equivalent.
\end{example}
We observe that, in general, the resulting \HornALCDat query has an exponential size compared to the original \HornALCHI query.
This blow-up arises in the second step and cannot be avoided and already occurs in the part of the translation borrowed from \cite{normal-form-horn-shiq}: when removing even a single axiom of form \eqref{rule:rhs-existential}, one may need exponentially-many axioms of form \eqref{rule:conjunction} to preserve equivalence of the rule queries.
Such pathological cases can be extracted from proofs of the ExpTime-hardness of basic reasoning tasks already in \HornALC, see \emph{e.g.}\ \cite{normal-form-horn-alc}.

\section{From \HornALCDat Queries to DL Automata}
\label{section:normalised-horn-alc-to-automata}

We introduce DL automata as finite-state devices closely related to regular tree automata, tailored to capture the expressive power of \HornALCDat queries.
We first formally introduce these in \Cref{definition:DL automaton} and then explain their semantics in \Cref{definition-run-automaton}.
Next, in \Cref{definition:reduction-dl-query-to-DL automaton}, we present a reduction from \HornALCDat queries to DL automata and show that this translation preserves the semantics of the input queries.
Finally, in \Cref{remark:DL automaton-properties}, we state some properties of the automata produced by this translation, which are applied in the following section to develop an equivalence-preserving reduction from a class of DL automata to PQs.

\begin{definition}
\label{definition:DL automaton}
We distinguish a special \emph{accepting state} \AcceptingState.
A \emph{binary transition} is an expression of the form \Transition{q}{\FP{S}}{Q}, where $q$ is a state, $\FP{S} \in \{\EmptyTransition\} \cup \Classes \cup \RoleExpressions$, and $Q$ is a set of states with $\vert Q \vert = 1$.
A \emph{multi-ary transition} is an expression of the form \Transition{q}{\EmptyTransition}{Q}, where $q$ is a state and $Q$ is a set of states with~$\vert Q \vert \geq 2$.
A \emph{transition} is either a binary transition or a multi-ary transition.
A \emph{DL automaton} is a tuple~$\langle s, Q, \InitialState, \TF \rangle$, where $s$ is a constant, $Q$ is a finite set of states, $\InitialState$ is a state, and $\TF$ is a \emph{transition function} over $Q$, that is, a finite set of transitions featuring only states in $Q$.
\end{definition}

Every DL automaton defines a boolean query and thereby partitions the set of all fact sets into those it accepts and those it rejects.
As for rule queries and PQs, this partition is formalized by the ``query mapping'' function $\QueryMapping$, defined below.
To characterize the accepted fact sets, we first introduce the notion of a run of a DL automaton, analogous to a run of a tree automaton on a tree.



\begin{definition}
\label{definition-run-automaton}
A \emph{run} of a DL automaton $\langle s, Q, \InitialState, \TF \rangle$ over a fact set $\FS$ from~$\langle p, e \rangle$, where~$p$ is a state and $e$ is a constant, is a triple $\langle N, E, L \rangle$ such that:\begin{itemize}
\item The pair $\langle N, E \rangle$ is a finite rooted directed tree. 
Moreover, $L$ is a labeling function that maps every node to a pair consisting of a state and a constant, and maps the root to $\langle p, e \rangle$.
\item For every node $n$ with $L(n) = \langle q, c \rangle$ that has exactly one child $m$ in the tree, there is a binary transition $\Transition{q}{\FP{S}}{\{q'\}} \in \TF$ satisfying all of the following implications:
\FirstItem If $\FP{S} = \EmptyTransition$, then $L(m) = \langle q', c \rangle$.
\SecondItem If $\FP{S}$ is a class, then $\FP{S}(c) \in \FS$ and $L(m) = \langle q', c \rangle$.
\ThirdItem If $\FP{S}$ is a role expression, then there is some constant $d$ such that $\FP{S}(c, d) \in \FS \cup \Invert(\FS)$ and $L(m) = \langle q', d \rangle$.
\item For every node $n$ with $L(n) = \langle q, c \rangle$ with at least two children $m_1, \ldots, m_k$, there is a multi-ary transition $\Transition{q}{\EmptyTransition}{\{q_1, \ldots, q_k\}} \in \TF$ such that $L(m_i) = \langle q_i, c \rangle$ for every $1 \leq i \leq k$.
\end{itemize}
Such a run is \emph{accepting} if all leaf labels in the tree feature the accepting state, and the constant $s$ occurs in \FS.

A DL automaton  $\DLAutomaton = \langle s, \InitialState,  Q, \TF \rangle$ \emph{accepts} a fact set \FS if there is an accepting run of \DLAut over \FS from $\langle \InitialState, s\rangle$.
Furthemore, let $\QueryMapping(\DLAutomaton)$ be the set of all fact sets accepted by \DLAutomaton.
Such an automaton is \emph{contained} in a rule query or a PQ~$\Query$ if $\QueryMapping(\DLAutomaton) \subseteq \QueryMapping(\Query)$; these two are \emph{equivalent} if $\QueryMapping(\DLAutomaton) = \QueryMapping(\Query)$.
\end{definition}

\renewcommand{\XSep}{1.2}
\renewcommand{\YSep}{0.35}

\begin{figure}[t]
\begin{tikzpicture}
\node[state] (qTrusted) at  (0*\XSep, 10*\YSep) {$q_\FP{Trust}$};
\node[state, accepting] (qAccept) at  (2*\XSep, 10*\YSep) {\AcceptingState};
\node[state] (qGate) at  (4*\XSep, 10*\YSep) {$q_\FP{Gate}$};
\node[state] (qSens) at  (2*\XSep, 3*\YSep) {$q_\FP{Sens}$};
\node[state] (qCrit) at  (0*\XSep, 0*\YSep) {$q_\FP{Crit}$};
\node[circle, fill, inner sep=0pt] (int) at (2*\XSep, 0*\YSep) {};
\node[state] (qOnl) at  (4*\XSep, 0*\YSep) {$q_\FP{Onl}$};

\path[->] (qSens) edge[out=225,in=180,looseness=8] node[rounded corners, draw, pos=0.5, fill=white, inner sep=2pt, yshift=3pt] {\FP{Link}} (qSens);

\path[->] (qOnl) edge[out=160,in=135,looseness=22] node[rounded corners, draw, pos=0.5, fill=white, inner sep=2pt, yshift=-8pt] {\FP{Link}} (qOnl);

\draw[->] ($(qTrusted.west)+(-\XSep*0.2,0)$) -- (qTrusted);
\draw[->] (qTrusted) -- (qAccept);
\node[rounded corners, draw, fill=white, inner sep=2pt] at ($(qTrusted)!0.5!(qAccept)$) {$\FP{Trust}$};

\draw[->] (qTrusted) -- (qCrit);
\node[rounded corners, draw, fill=white, inner sep=2pt] at ($(qTrusted)!0.5!(qCrit)$) {$\FP{Acc}^-$};

\path[-] (qCrit) edge (int);
\path[->] (int) edge (qSens);
\path[->] (int) edge (qOnl);
\node[rounded corners, draw, fill=white, inner sep=2pt] at ($(qCrit)!0.5!(qOnl)$) {\EmptyTransition};
\draw[->] (qCrit) -- (qAccept);
\node[rounded corners, draw, fill=white, inner sep=2pt] at ($(qCrit)!0.5!(qAccept)$) {\FP{Crit}};

\draw[->] (qSens) -- (qAccept);
\node[rounded corners, draw, fill=white, inner sep=2pt] at ($(qSens)!0.5!(qAccept)$) {\FP{Sens}};

\draw[->] (qOnl) -- (qGate);
\node[rounded corners, draw, fill=white, inner sep=2pt] at ($(qOnl)!0.5!(qGate)$) {\EmptyTransition};
\draw[->] (qOnl) -- (qAccept);
\node[rounded corners, draw, fill=white, inner sep=2pt] at ($(qOnl)!0.5!(qAccept)$) {\FP{Onl}};

\draw[->] (qGate) -- (qAccept);
\node[rounded corners, draw, fill=white, inner sep=2pt] at ($(qGate)!0.5!(qAccept)$) {\FP{Gate}};

\node[circle, fill, inner sep=1.2pt, label=right:{$n_1 : \langle q_\FP{Trust}, \textit{alice} \rangle $}] (n1) at (6.75*\XSep, 10*\YSep) {};
\node[circle, fill, inner sep=1.2pt, label=right:{$n_2 : \langle q_\FP{Crit}, c_1 \rangle $}] (n2) at (6.75*\XSep, 8*\YSep) {};
\node[circle, fill, inner sep=1.2pt, label=right:{$n_3^r : \langle q_\FP{Sens}, c_1 \rangle $}] (n3l) at (7.75*\XSep, 6*\YSep) {};
\node[circle, fill, inner sep=1.2pt, label=right:{$n_4^r : \langle q_\FP{Sens}, c_2 \rangle $}] (n4l) at (7.75*\XSep, 4*\YSep) {};
\node[circle, fill, inner sep=1.2pt, label=right:{$n_5^r : \langle q_\FP{Sens}, c_3 \rangle $}] (n5l) at (7.75*\XSep, 2*\YSep) {};
\node[circle, fill, inner sep=1.2pt, label=right:{$n_6^r : \langle \AcceptingState, c_3 \rangle $}] (n6l) at (7.75*\XSep, 0*\YSep) {};
\node[circle, fill, inner sep=1.2pt, label={[below right]{$n_3^\ell : \langle q_\FP{Onl}, c_1 \rangle$}}] (n3r) at (5.75*\XSep, 6*\YSep) {};
\node[circle, fill, inner sep=1.2pt, label=right:{$n_4^\ell : \langle q_\FP{Onl}, c_4 \rangle $}] (n4r) at (5.75*\XSep, 4*\YSep) {};
\node[circle, fill, inner sep=1.2pt, label=right:{$n_5^\ell : \langle q_\FP{Gate}, c_4 \rangle $}] (n5r) at (5.75*\XSep, 2*\YSep) {};
\node[circle, fill, inner sep=1.2pt, label=right:{$n_6^\ell : \langle \AcceptingState, c_4 \rangle $}] (n6r) at (5.75*\XSep, 0*\YSep) {};

\draw[->] (n1) -- (n2);
\draw[->] (n2) -- (n3l);
\draw[->] (n2) -- (n3r);
\draw[->] (n3l) -- (n4l);
\draw[->] (n3r) -- (n4r);
\draw[->] (n4l) -- (n5l);
\draw[->] (n4r) -- (n5r);
\draw[->] (n5l) -- (n6l);
\draw[->] (n5r) -- (n6r);

\end{tikzpicture}
\caption{Consider the DL automaton $\langle \textit{alice}, Q, q_\FP{Trust}, \TF \rangle$, where \TF is the transition function over $Q$ depicted on the left, and the only accepting run of this automaton from~$\langle q_\FP{Trust}, \textit{alice} \rangle$ over the fact set \FS introduced in \Cref{example:introduction-ontology} depicted on the right.}
\label{figure:example-automata-accepting-run}
\end{figure}
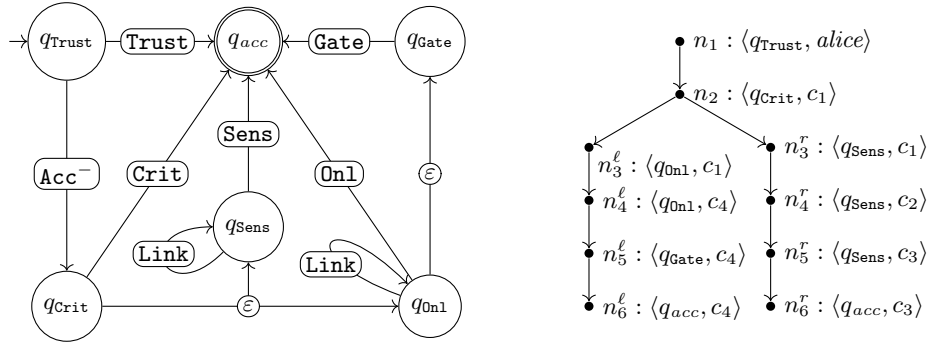

\begin{example}
In \Cref{figure:example-automata-accepting-run}, we present a DL automaton together with one of its accepting runs over the fact set \FS introduced in \Cref{example:introduction-ontology}.
We will later argue that this automaton is equivalent to the \HornALCDat query $\langle \TBox_\textit{NF}, \FP{Trust}(\textit{alice}) \rangle$ introduced in \Cref{example:normalise-away-existentials}; for now, the figure should help the reader understand why this automaton accepts this fact set.
\end{example}

In the following definition, we present a reduction from \HornALCDat queries to DL automata that preserves equivalence (cf. \Cref{theorem:correspondence-normalised-horn-alc-automata}).

\begin{definition}
\label{definition:reduction-dl-query-to-DL automaton}
For every $\FP{C} \in \Classes$, let $q_\FP{C}$ be a fresh state unique for $\FP{C}$.
For a $\HornALCDat$ query $\Query = \langle \RS, \FP{D}(s) \rangle$, let $\ToDLA(\Query) = \langle s, Q, q_\FP{D}, \TF \rangle$ denote the DL automaton where $Q$ is the set of states containing \AcceptingState and $q_\FP{C}$ for every class $\FP{C}$ occurring in \Query, and \TF is the transition function containing all of the following:
\begin{itemize}
\item For every class $\FP{C}$ occurring in $\Query$, add $\Transition{q_\FP{C}}{\FP{C}}{\{\AcceptingState\}} \in \TF$.
\item For every rule such as $\bigwedge^n_{i = 1} \FP{A}_i(x) \to \FP{B}(x) \in \RS$, add $\Transition{q_\FP{B}}{\EmptyTransition}{\{q_\FP{A_1}, \ldots, q_\FP{A_n}\}} \in \TF$.
\item For every rule such as $\FP{R}(x, y) \wedge \FP{A}(y) \to \FP{B}(x) \in \RS$, add $\Transition{q_\FP{B}}{\FP{R}}{\{q_\FP{A}\}} \in \TF$.
\item For every rule such as $\FP{A}(x) \wedge \FP{R}(x, y) \to \FP{B}(y) \in \RS$, add $\Transition{q_\FP{B}}{\FP{R}^-}{\{q_\FP{A}\}} \in \TF$.
\item For every rule such as $\to \FP{B}(x) \in \RS$, add $\Transition{q_\FP{B}}{\EmptyTransition}{\{\AcceptingState\}} \in \TF$.
\end{itemize}
\end{definition}

\begin{example}
Consider the rule set $\TBox_\textit{NF}$ and the fact set \FS introduced in \Cref{example:introduction-ontology}.
Observe that the automaton $\HornALCToDLAutomaton(\TBox_\textit{NF}, \FP{Trust}(\textit{alice}))$, described in \Cref{figure:example-automata-accepting-run}, accepts \FS, as witnessed by the accepting run depicted in that same figure.
By contrast, the automaton $\HornALCToDLAutomaton(\TBox_\textit{NF}, \FP{Trust}(\textit{bob}))$ does not accept \FS, as there is no accepting run from $\langle q_\FP{Trust}, \textit{bob} \rangle$.
\end{example}

\begin{restatable}{theorem}{CorrespondenceNormHornALCAutomata}
\label{theorem:correspondence-normalised-horn-alc-automata}
A \HornALCDat query $\Query$ is equivalent to the DL automaton $\ToDLA(\Query)$.
\end{restatable}
\begin{proof}[Sketch]
The theorem follows from the fact that, for every fact set \FS and fact $\FP{B}(b)$, the KB $\langle \RS, \FS \rangle$ entails $\FP{B}(b)$ if and only if there is an accepting run of $\HornALCToDLAutomaton(\RS, \FP{B}(b))$ over \FS starting from $\langle q_\FP{B}, b \rangle$.
The ``if'' direction is proved by induction on the height of the accepting run, while the ``only if'' direction is established by induction on the chase depth of $\FP{B}(b)$ with respect to $\langle \RS, \FS \rangle$.
\end{proof}

To conclude the section, the following remark states several properties satisfied by every DL automaton produced by the reduction of \Cref{definition:reduction-dl-query-to-DL automaton}.

\begin{remark}
\label{remark:DL automaton-properties}
Consider a \HornALCDat query $\Query$ and the DL automaton $\ToDLA(\Query) = \langle s, Q, \InitialState, \TF \rangle$.
Then, all of the following hold:
\begin{enumerate}
\item There are no transitions originating from the accepting state; that is, $\TF$ does not contain transitions of the form $\Transition{\AcceptingState}{\FP{S}}{P}$.
\label{definition:DL automaton-accepting-state-sink}
\item Every (binary) transition labelled with a class points to the accepting state; that is, for every $\Transition{q}{\FP{C}}{\{p\}} \in \TF$ with $\FP{C} \in \Classes$, we have $p = \AcceptingState$.
\label{definition:DL automaton-class-labelled-to-accepting}
\item Every non-final state directly reaches the accepting state through a binary transition; that is, for every $q \in Q \setminus \{\AcceptingState\}$, there is some transition of the form $\Transition{q}{\FP{S}}{\{\AcceptingState\}} \in \TF$.
\label{definition:DL automaton-all-states-reach-accept}
\end{enumerate}
\end{remark}

Properties~\ref{definition:DL automaton-accepting-state-sink} and \ref{definition:DL automaton-class-labelled-to-accepting} together ensure that transitions labelled with classes are considered only at the very end of an accepting run.
This, in turn, simplifies the form of the regular expressions used in the following section to encode DL automata, since these only need to involve role expressions.
Property~\ref{definition:DL automaton-all-states-reach-accept} ensures that every state admits a path to the accepting state.
If not for this restriction, a DL automaton may not accept any fact set.
Such an automaton cannot be translated into a PQ, since the empty query is not a PQ under our definition.

\section{From Stratified DL Automata to PQs}
\label{section:automata-to-acyclic-pqs}

In this section, we present an equivalence-preserving reduction from a class of DL automata to PQs.
More precisely, this reduction applies only to stratified automata, which we formally introduce in the following definition.

\begin{definition}
\label{definition:stratified-DL automaton}
Consider a DL automaton $\DLAut = \langle s, Q, \InitialState, \TF \rangle$.
\begin{itemize}
\item For states $q$ and $p$, we write $\BinTransRel{\DLAut}{q}{p}$ if $\TF$ contains a binary transition of the form $\Transition{q}{\FP{S}}{\{p\}}$, and $\MultiTransRel{\DLAut}{q}{p}$ if $\TF$ contains a multi-ary transition of the form~$\Transition{q}{\EmptyTransition}{P}$ with $p \in P$.
\item We define $\BothTransRel{\DLAut}{}{}$ as the union of $\BinTransRel{\DLAut}{}{}$ and $\MultiTransRel{\DLAut}{}{}$, and we define $\TRTransRel{\DLAut}{}{}$ as the reflexive transitive closure of $\BothTransRel{\DLAut}{}{}$.
\item For states~$q$ and~$p$, we write $\InHigherStrata{\DLAut}{q}{p}$ if there are $q', q'', p', p'' \in Q$ and $\FP{E} \in \RoleExpressions$ with $\TRTransRel{\DLAut}{q}{q'}$, $\Transition{q'}{\FP{E}}{\{q''\}} \in \TF$, $\TRTransRel{\DLAut}{q''}{p'}$, $\MultiTransRel{\DLAut}{p'}{p''}$, and $\TRTransRel{\DLAut}{p''}{p}$.
\end{itemize}

A DL automaton $\DLAut$ is \emph{stratified} if the relation $\InHigherStrata{\DLAut}{}{}$ is acyclic, and if it satisfies the conditions in \Cref{remark:DL automaton-properties}.
If $\DLAut$ is stratified, we define $\StrFun{\DLAut} : Q \to \mathbb{N}$ as the function mapping a state $q \in Q$ to $1$ if $q$ has no successors with respect to $\InHigherStrata{\DLAut}{}{}$, and to $\FF{Max}(\{\StrFun{\DLAut}(p) \mid \InHigherStrata{\DLAut}{q}{p}\}) + 1$ otherwise.
\end{definition}


\begin{example}
\label{example:strafication}
Consider the DL automaton $\DLAut = \langle \textit{alice}, Q, q_\FP{Trust}, \TF \rangle$ described in \Cref{figure:example-automata-accepting-run}.
For instance, we have $\InHigherStrata{\DLAut}{q_\FP{Trust}}{q_\FP{Gate}}$, since all of the following hold:
\begin{align*}
\TransitiveReflexiveTransitionRelation{\DLAut}{q_\FP{Trust}}{q_\FP{Trust}}
\qquad \Transition{q_\FP{Trust}}{\FP{Acc}^-}{\{q_\FP{Crit}\}} \in \TF
\qquad \TransitiveReflexiveTransitionRelation{\DLAut}{q_\FP{Crit}} {q_\FP{Crit}}
\\
\MultiTransitionRelation{\DLAut}{q_\FP{Crit}}{q_\FP{Onl}}
\qquad \TransitiveReflexiveTransitionRelation{\DLAut}{q_\FP{Onl}}{q_\FP{Gate}}
\end{align*}
Observe that \DLAut is stratified, since every cycle in its transition function involves only binary transitions.
\end{example}

Our reduction from automata to PQs, formally introduced in \Cref{section:DL automata-to-pqs-inductive-definition}, proceeds by associating with every state $q$ of a stratified DL automaton~\DLAut a PQ describing the behaviour of \DLAut when started in~$q$.
This association is defined inductively: for states in stratum~$1$, the corresponding PQ is defined directly, whereas for states in stratum~$k \geq 2$, it is defined using the PQs already associated with states in strictly lower strata.
The role of stratification is to guarantee that this recursive construction is well defined. 
Indeed, the relation~\InHigherStrata{\DLAut}{}{} captures precisely those dependencies between states that arise when the translation encounters a path containing a binary transition labelled with a role expression and, possibly after intermediate transitions, a multi-ary transition.
In this situation, our inductive construction branches out and introduces additional existentially quantified variables, thereby producing formulas such as the third disjunct of the PQ in \Cref{example:introduction-query}.
If cycles over \InHigherStrata{\DLAut}{}{} were allowed, the translation could follow the same dependencies indefinitely, repeatedly expanding the formula and introducing new existential variables, thus failing to produce a finite PQ.

We point out that in \Cref{definition:stratified-DL automaton} we require stratified automata to satisfy the conditions of \Cref{remark:DL automaton-properties}, since these are necessary to prove some of the results that follow.
This is without loss of generality when considering DL automata corresponding to \HornALCDat queries since, as explained in the previous section, all such automata satisfy these conditions.
We define our reduction from automata to PQs in \Cref{section:DL automata-to-pqs-inductive-definition}; its definition relies on a preliminary translation for DL automata without multi-ary transitions, presented in \Cref{section:nfas-to-pqs}.

\subsection{DL Non-Deterministic Finite Automata}
\label{section:nfas-to-pqs}

In the following subsection, we present a reduction that transforms DL automata into equivalent PQs, that is, to positive FO formulas featuring regular expressions.
To define such a transformation, we leverage existing automata-theoretic techniques defined for automata without multi-ary transitions.

\begin{definition}
\label{definition:nfa}
A DL \emph{non-deterministic finite automaton} (NFA) \DLAut is a DL automaton without multi-ary transitions.
Given states~$q$ and $p$, constants $s$ and $t$, and a fact set \FactSet, the automaton \DLAut \emph{transitions} from~$\langle q, s \rangle$ to $\langle p, t \rangle$ in \FactSet if there is a run of \DLAut over \FactSet from $\langle q, s \rangle$ whose (unique) leaf is labelled with $\langle p, t \rangle$.
\end{definition}

Observe that, by \Cref{definition-run-automaton}, in every run of a DL automaton without multi-ary transitions, each node has at most one child; hence, such a run has exactly one leaf.
It is well known from classical automata theory that the set of all paths induced by such an automaton can be characterized by a regular expression \cite{intro-automata-theory}.

\begin{lemma}
\label{lemma-nfa-transition}
Consider a DL NFA $\NFA$ that satisfies \Cref{remark:DL automaton-properties},\footnote{Recall that NFAs are DL automata by \cref{definition:nfa}.} and some states $q$ and $p$ of \NFA such that $p$ is different from $\AcceptingState$.
Then, there is a regular expression~$\RegularExpressionNFAStateToState{q}{p}{\NFA}$ such that, for every fact set~$\FactSet$ and constants $s$ and $t$, the NFA $\NFA$ can transition from~$\langle q, s \rangle$ to $\langle p, t \rangle$ in~$\FactSet$ if and only if $\FactSet \models \RegularExpressionNFAStateToState{q}{p}{\NFA}(s, t)$.
\end{lemma}
\begin{proof}
By Kleene’s theorem \cite{kleene-theorem}, there is a regular expression such as $\RegularExpressionNFAStateToState{q}{p}{\NFA}$ over role expressions such that $\Language(\RE_{q \triangleright p}^\DLAut)$ is the set of all words labelling the paths from $q$ to $p$ in~\NFA.
This applies because no path from $q$ to $p$ visits the accepting state~\AcceptingState since $p \neq \AcceptingState$ and \AcceptingState is a sink in~\NFA by Item~\ref{definition:DL automaton-accepting-state-sink} of \Cref{remark:DL automaton-properties}.
Therefore, these paths contain no classes, since by Item~\ref{definition:DL automaton-class-labelled-to-accepting} of the same remark, all binary transitions of \NFA\ labelled with classes have \AcceptingState\ as target.
\end{proof}

The previous auxiliary result allows us to readily define an equivalence-preserving reduction from NFAs to PQs.

\begin{definition}
\label{definition:reduction-nfa-to-pq}
Given an NFA $\NFA = \langle s, Q, \InitialState, \TF \rangle$ that satisfies \Cref{remark:DL automaton-properties}, a state~$q$ of \NFA, and a term~$t$, we define $\NFAToPostiveQuery(\NFA, q, t)$ as the following PQ:
\begin{itemize}
\item If $q$ is the accepting state \AcceptingState, then this PQ is equal to $\top$.
\item Otherwise, this PQ is the disjunction containing, for every $p \in Q \setminus \{\AcceptingState\}$ with $\TransitiveReflexiveTransitionRelation{\NFA}{q}{p}$, all of the following:
\FirstItem If $\Transition{p}{\FP{C}}{\{\AcceptingState\}} \in \NFA$ for some $\FP{C} \in \Classes$, add the disjunct $\exists w . \big[ \RegularExpressionNFAStateToState{q}{p}{\NFA}(t, w) \wedge \FP{C}(w) \big]$ with $w$ a fresh variable.
\SecondItem If $\Transition{p}{\EmptyTransition}{\{\AcceptingState\}} \in \NFA$, add the disjunct $\exists w . \RegularExpressionNFAStateToState{q}{p}{\NFA}(t, w)$ with $w$ a fresh variable.
\end{itemize}
\end{definition}

Since \NFA above satisfies \Cref{remark:DL automaton-properties}, every state distinct from the accepting state reaches it via a binary transition annotated with a class.
Therefore, the corresponding PQ is either $\top$ or a non-empty disjunction.

\begin{restatable}{lemma}{lemmatransitioncrpq}
\label{lemma-transition-crpq}
Consider a DL NFA $\NFA$ that satisfies \Cref{remark:DL automaton-properties}, a fact set \FS, a state~$q$ of \NFA different from \AcceptingState, and a constant~$s$ in \FS.
There is an accepting run of \NFA over \FS starting at~$\langle q, s \rangle$ if and only if $\FS$ satisfies $\NFAToPostiveQuery(\NFA, q, s)$.
\end{restatable}

The previous result yields an equivalence-preserving reduction from NFAs to PQs illustrated in the following example.

\begin{example}
\label{example:reduction-nfa-to-pq}
Consider the DL automaton $\langle \textit{alice}, Q, q_\FP{Trust}, \TF \rangle$ depicted \Cref{figure:example-automata-accepting-run} and the NFA $\NFA = \langle s, Q, q_\FP{Onl}, \TF' \rangle$ where $\TF'$ is the set of all binary transitions in~$\TF$.
By \Cref{lemma-transition-crpq}, this NFA is equivalent to the following PQ:
$$\NFAToPostiveQuery(\NFA, q_\FP{Onl}, s)= \exists u_1 . \big[ \FP{Link}^*(s, u_1) \wedge \FP{Onl}(u_1) \big] \vee \exists u_2 . \big[\FP{Link}^*(s, u_2) \wedge \FP{Gate}(u_2)\big]$$
\end{example}

In the following subsection, we use the previously discussed reduction as a subroutine in a more general translation applicable to stratified DL automata, which may contain multi-ary transitions.
To apply this subroutine, we define the NFA reachable from a given state in a DL automaton.

\begin{definition}
Given a DL automaton $\DLAut$ and a state $q$ of \DLAut, we define $\FF{NFA}(\DLAut, p)$ as the NFA obtained from $\DLAut$ by performing the following modifications in sequence: remove all multi-ary transitions, remove all states that are not reachable from $p$ through the remaining (binary) transitions, and remove all (binary) transitions that mention a state removed in the previous step.
\end{definition}

\subsection{An Inductive Reduction for Stratified DL Automata}
\label{section:DL automata-to-pqs-inductive-definition}

Before presenting our reduction, we define an extension of transition functions that allows adjacent \EmptyTransition-transitions to be handled in single step.

\begin{definition}
For a transition function $\TF$ over a set of states~$Q$, let $\TF^*_\EmptyTransition$ be the (minimal) set of multi-ary transitions such that:
For every $q \in Q$, we have $\MultiaryTrans{q}{q} \in \TF^*_\EmptyTransition$.
For every~$\Transition{q}{\EmptyTransition}{P} \in \delta_\EmptyTransition^*$ and $\Transition{p}{\EmptyTransition}{P'} \in  \TF$ with $p \in P$, we have $\Transition{q}{\EmptyTransition}{(P\setminus \{p\}) \cup P'} \in  \TF^*_\EmptyTransition$.
\end{definition}

Every DL automaton $\langle s, Q, \InitialState, \TF \rangle$ is equivalent to the extended automaton $\langle s, Q, \InitialState, \TF \cup \TF^*_\EmptyTransition \rangle$, since every transition in $\TF^*_\EmptyTransition$ can be simulated by a finite sequence of \EmptyTransition-transitions already present in $\TF$.
Nevertheless, this extension allows every combination of states reachable via \EmptyTransition-transitions to be reached directly in a single step, thus avoiding potential cyclic dependencies in our reduction.

In the following definition, we formally introduce our equivalence-preserving reduction, which is admittedly concise and technically involved.
We suggest that the reader first skim through it, then consult the examples that follow, and only then return to the definition for a clearer understanding.

\begin{definition}
\label{definition:reduction-stratified-DL automaton-to-pq}
Given a stratified DL automaton $\DLAut = \langle s, Q, \InitialState, \TF \rangle$, a state~$q \in Q$, a term $t$, we inductively define $\StratDLAToPostiveQuery(\DLAut, q, t)$ as the following PQ:
\begin{align*}
&\bigvee\nolimits_{q \rightsquigarrow_\EmptyTransition P \in \TF^*_\EmptyTransition} \bigwedge\nolimits_{p \in P} \Big(\NFAToPostiveQuery(\NFA, p, t)~\vee \\
&\quad \bigvee_{\substack{\forall p', p_1, \ldots, p_n \in Q \textit{ such that } \TRTransRel{\NFA}{p}{p'}, \\ \Transition{p'}{\EmptyTransition}{\{p_1,\ldots, p_n\}},~\InHigherStrata{\DLAut}{p}{p_1}, \ldots, \text{ and } \InHigherStrata{\DLAut}{p}{p_n}}} \exists w. \big[\RegularExpressionNFAStateToState{p}{p'}{\NFA}(t, w) \wedge \bigwedge\nolimits_{i = 1}^n \StratDLAToPostiveQuery(\DLAut, p_i, w)\big] \Big)
\end{align*}
In the above, $\NFA$ is the NFA $\FF{NFA}(\DLAut, q)$ and $w$ is a fresh variable.

In the following, we will sometimes write $\StratDLAToPostiveQuery(\DLAut)$ to refer to $\StratDLAToPostiveQuery(\DLAut, \InitialState, s)$.
\end{definition}

Observe that the above reduction is well defined since, for every stratified DL automaton~\DLAut, the relation \InHigherStrata{\DLAut}{}{} is acyclic.
Hence, we can first compute the PQs corresponding to all states in the lowest stratum, that is, states without successors with respect to \InHigherStrata{\DLAut}{}{}.
We then proceed stratum by stratum, computing the subsequent PQs by relying on those obtained previously.
The following example illustrates both the base case and the inductive step of our construction.

\begin{example}
Consider the DL automaton $\DLAut$ from \Cref{figure:example-automata-accepting-run} and a constant $s$.
To illustrate the base case of our construction, we consider the state $q_\FP{Crit}$, which belongs to the lowest stratum of \DLAut, and present the corresponding PQ.
\begin{align*}
\StratDLAToPostiveQuery(\DLAut, q_\FP{Crit}, s) = ~& \big(\NFAToPostiveQuery(\NFA, q_\FP{Sens}, s) \wedge \NFAToPostiveQuery(\NFA, q_\FP{Onl}, s)\big)\, \vee \notag \\
\big(\NFAToPostiveQuery(\NFA, &\, q_\FP{Sens}, s) \wedge \NFAToPostiveQuery(\NFA, q_\FP{Gate}, s)\big) \, \vee \, \NFAToPostiveQuery(\NFA, q_\FP{Crit}, s) \\[0.75ex]
\NFAToPostiveQuery(\NFA, q_\FP{Crit}, s) = ~&\FP{Crit}(s) \\[0.75ex]
\NFAToPostiveQuery(\NFA, q_\FP{Gate}, s) = ~&\FP{Gate}(s) \\[0.75ex]
\NFAToPostiveQuery(\NFA, q_\FP{Sens}, s) = ~&\exists u_3 . \big[\FP{Link}^*(s, u_3) \wedge \FP{Sens}(u_3)\big]
\end{align*}
The formula $\NFAToPostiveQuery(\NFA, q_\FP{Onl}, s)$ was already given in \Cref{example:reduction-nfa-to-pq}.
To illustrate the inductive step of our construction, we consider the state $q_\FP{Trust}$, which belongs to stratum~$2$ of \DLAut, and present the corresponding PQ.
\begin{align*}
\StratDLAToPostiveQuery(\DLAut, q_\FP{Trust}, s) = ~&\FP{Trust}(s) \, \vee \, \exists u_4 . \big[\FP{Acc}^-(s, u_4) \wedge \FP{Crit}(u_4)\big]\, \vee \, \notag \\
\exists u_5 . \big[\FP{Acc}^-(s, u_5) &\wedge \StratDLAToPostiveQuery(\DLAut, q_\FP{Sens}, u_5) \wedge \StratDLAToPostiveQuery(\DLAut, q_\FP{Onl}, u_5)\big]
\end{align*}
\end{example}

\begin{restatable}{theorem}{CorrespondenceAutomataCTRPQ}
\label{theorem:correspondence-automata-crpq}
A stratified DL automaton $\DLAutomaton$ is equivalent to $\StratDLAToPostiveQuery(\DLAutomaton)$.
\end{restatable}

\begin{proof}[Sketch]
The theorem follows from the fact that, for every fact set \FS and fact $\FP{B}(b)$, the automaton \DLAut has an accepting run from $\langle q_\FP{B}, b \rangle$ over \FS if and only if \FS entails $\StratDLAToPostiveQuery(\DLAut, q_\FP{B}, b)$.
Both directions are proven by induction on the stratum of the state $q_\FP{B}$.
\end{proof}

\section{From Positive Queries to GQL Queries}
\label{section:ayclicic-c2rpqs-to-gql}

In this section, we show that PQs can be expressed in the Graph Query Language (GQL) \cite{sigmod22:gql,icdt23:digest-of-gql}.
Namely, we present an equivalence-preserving mapping from PQs to a query language that is already known to be expressible in GQL \cite{pods23:uc2rpq-to-gql}.

\begin{definition}
A \emph{union of conjunctive two-way regular path queries (UC2RPQ)} is a PQ of the form $\bigvee\nolimits_{i = 1}^n \exists \vec{y}_i . \Formula_i$ where all $\Formula_i$ are non-empty atom conjunctions.
\end{definition}

\begin{definition}
\label{definition:pq-to-uc2rpq}
Given a PQ \PQ, we define $\FF{UC2RPQ}(\PQ)$ as a PQ obtained by exhaustively applying the following replacement rules:
\begin{enumerate}
\item Replace a subformula of the form $\exists w . (\Formula \vee \FormulaAux)$ with $(\exists w . \Formula) \vee (\exists v . \FormulaAux[w / v])$, where $v$ is a fresh variable and $\FormulaAux[w / v]$ is the formula that results from replacing all occurrences of $w$ in \FormulaAux with $v$.\label{nf-rule:exist-inside-disjunction}
\item Replace a subformula of the form $(\Formula_1 \wedge (\Formula_2 \vee \Formula_3))$ with $(\Formula_1 \wedge \Formula_2) \vee (\Formula_1 \wedge \Formula_3)$.\label{nf-rule:wedge-inside-disjunction}
\item Replace a subformula of the form $(\exists w . \Formula) \wedge (\exists v . \FormulaAux)$ with $\exists w, v . (\Formula \wedge \FormulaAux)$.\label{nf-rule:exists-inside-conjunction}
\end{enumerate}
\end{definition}

The previous translation is included for completeness and is not an original contribution, as it relies on standard transformations \cite[Chapter~5]{foundations-of-databases-95}.
\begin{lemma}
For every PQ \PQ, the queries \PQ and $\FF{UC2RPQ}(\PQ)$ are equivalent, and $\FF{UC2RPQ}(\PQ)$ is a UC2RPQ.
\end{lemma}
\begin{proof}
One can prove the equivalence of the above formulas by induction on the sequence of formulas obtained by applying the replacement rules from \Cref{definition:pq-to-uc2rpq} in the computation of $\FF{UC2RPQ}(\PQ)$.
Since the latter formula is obtained by exhaustively applying Rules~\ref{nf-rule:exist-inside-disjunction} and \ref{nf-rule:wedge-inside-disjunction}, disjunctions occur only at the outermost level.
Moreover, by the exhaustive application of Rule~\ref{nf-rule:exists-inside-conjunction}, conjunctions occur only at the innermost level.
Therefore, $\FF{UC2RPQ}(\PQ)$ is a UC2RPQ.
\end{proof}

The following is a corollary of the previous lemma and the equivalence-preserving reduction of unions of C2RPQs to GQL \cite[Theorem~11]{pods23:uc2rpq-to-gql}.
\begin{theorem}
\label{theorem:correspondence-pq-to-gql}
There is an equivalence-preserving reduction from PQs to GQL.
\end{theorem}

\section{Related Work}
\label{section:related-work}

In the line of work on rewriting DL queries, a cornerstone is the first-order rewritability of rule queries of the DL-Lite family, allowing translation into unions of Select-Project-Join SQL queries \cite{calvaneseetal:dllite}.
It is worth noting that DL automata associated to (the positive versions) of DL-Lite rule queries are always stratified.
Indeed, the popular fragments of DL-Lite (DL-Lite$_{core}$, DL-Lite$_\mathcal{R}$) forbid axioms of the form \eqref{rule:conjunction} for $n \geq 2$.
By unraveling the successive steps exposed in \Cref{section:horn-alc-to-normalised-horn-alc}, it is not difficult to verify that axioms of the form \eqref{rule:conjunction}, for $n \geq 2$, are not introduced throughout the normalization procedure.
Multi-ary transitions are therefore absent from the corresponding DL automata, which guarantees the existence of a stratification with a single stratum.

Closer to our setting, rewriting of rule queries from (extensions of) $\EL$ into queries based upon regular expressions has also been studied.
In \cite{linear-EL}, rule queries of linear $\EL$ are rewritten into regular-path queries, and this work was later extended to the so-called \emph{harmless linear $\ELHI$} \cite{harmless-elhi} using UC2RPQs as a target query language.
Harmless linear $\ELHI$ restricts axioms of the form \eqref{rule:conjunction} to the case of $n = 1$, as in DL-Lite, and limits the interaction of pairs of roles $\FP{S_1}, \FP{S_2}$ whenever there exists a third role $\FP{R}$ such that both $\FP{R}(x, y) \to \FP{S_1}(x, y)$ and $\FP{R}(x, y) \to \FP{S_2}(y, x)$ are entailed by the set of rules.
Already due to the restriction on conjunction, \Cref{example:introduction-query} falls outside of the rule queries that their approach can rewrite.
The nature of the restrictions in harmless linear $\ELHI$ makes a more formal comparison with our approach difficult.
We have not found an example of a harmless linear $\ELHI$ rule query that does not have an associated stratified DL automaton.

More recently, and also motivated by the navigational features of GQL, \emph{quasi-linear $\mathcal{ELH}$ with limited inverses} ($\quasilinELHI$) was introduced and has been shown to enjoy rewriting into UC2RPQs \cite{eswc25-rew-quasilinear}.
The definition of this latter extension of $\EL$ is characterized by a notion of \emph{locality}. 
Typically, our axiom $\FP{Crit}(x) \wedge \FP{Acc}(x, y) \to \FP{Trust}(y)$ makes the class $\FP{Crit}$ non-local, and $\quasilinELHI$ forbids non-local classes from appearing in the head of an axiom of form \ref{rule:conjunction}.
As a consequence, an axiom like $\FP{Sens}(x) \wedge \FP{Onl}(x) \to \FP{Crit}(x)$ is not allowed in \quasilinELHI.
This shows that Example~\ref{example:introduction-query} is not a \quasilinELHI rule query, while being rewritable with our approach since its corresponding DL automaton is stratified (see Example~\ref{example:strafication}).
We speculate that by examining closely how the notion of (non-)locality is preserved through the different normalization steps exposed in \Cref{section:horn-alc-to-normalised-horn-alc}, one can establish that every $\quasilinELHI$ rule query has an associate stratified DL automaton.
Together with the previous remark regarding Example~\ref{example:introduction-query}, such a verification would guarantee that our work strictly subsumes theirs.

A very recent work \cite{DBLP:conf/dlog/ArpasiBO26} has also pinpointed a subclass of \ELHI queries that supports rewriting into nested 2RPQs, another fragment of GQL.
The proposed subclass, namely $\ELI^{\bot}_{\preceq}$, relies on the existence of a stratification of class and role expressions, constrained by the rule set.
This stratification notably limits recursion to rules with shape $\FP{R}(x, y) \land \FP{A}(y) \to \FP{A}(x)$ or $\FP{A}(x) \land \FP{R}(x, y) \to \FP{A}(y)$.
As a consequence, their approach cannot rewrite the query $\langle \{\FP{R}(x, y) \land \FP{A}(y) \to \FP{B}(x), \FP{R}(x, y) \land \FP{B}(y) \to \FP{A}(x)\}, \FP{A}(a)\rangle$, whose corresponding DL automaton is trivially stratified.
Their work also seemingly varies from ours as the nested 2RPQs they consider support \emph{tests} along the 2RPQs.
However, we did not find an example of an $\ELI^{\bot}_{\preceq}$ query that is not already rewritable using our approach; a natural candidate exploiting the presence of tests in their formalism would have been the rule query $\langle \{ \FP{A}(x) \land \FP{B}(x) \to \FP{C}(x), \FP{R}(x, y) \land \FP{C}(y) \to \FP{A}(x) \}, \FP{C}(c) \rangle$, but both approaches fail to rewrite this query.

\section{Conclusions and Future Work}
\label{section:work-conclusions-future}

As our main contribution, we define an equivalence-preserving procedure that transforms a class of \HornALCHI queries into equivalent GQL formulas, proceeding in several steps.
First, an input \HornALCHI query is transformed into a \HornALCDat query (\Cref{theorem:normalization-is-correct}), which is then translated into a DL automaton (\Cref{theorem:correspondence-normalised-horn-alc-automata}).
If the resulting automaton is stratified, it can be further translated into a PQ (\Cref{theorem:correspondence-automata-crpq}), and subsequently into a GQL query (\Cref{theorem:correspondence-pq-to-gql}).

An immediate direction for future work is to extend our translation to handle other constructors available in standard DL languages.
For instance, our technique can be readily adapted to incorporate disjointness rules of the form $\FP{A_1}(x) \wedge \ldots \wedge \FP{A_n}(x) \to \bot$, which are typically allowed in Horn DLs \cite{normal-form-horn-alc}.
When computing the rewriting of a \HornALCHI query $\langle \RS, \FP{A}(a) \rangle$ in the presence of such rules, we may treat $\bot$ as a normal class and instead compute a GQL rewriting for the rule query~$\langle \RS, \FP{A}(a) \vee \exists w . \bot(w) \rangle$; similar ideas have already been considered in e.g. \cite{DBLP:journals/ws/CaliGL12}.
This approach is correct since, for every fact set $\FS$, the following hold: the KB~$\langle \RS, \FS \rangle$ is unsatisfiable if and only if it entails $\exists w . \bot(w)$; and $\langle \RS, \FS \rangle$ entails~$\FP{A}(a)$ if and only if it is unsatisfiable, or $\langle \RS_{\centernot\bot}, \FS \rangle \models \FP{A}(a)$ with $\RS_{\centernot\bot}$ the subset of $\RS$ without syntactic occurrences of~$\bot$.
Extending our translation to other DL constructors, such as role chains, is likely to be more challenging.
Nevertheless, complex roles appear amenable to our approach, since these closely correspond to regular languages \cite{sriq-rewriting-into-datalog,removing-complex-roles} and could in principle be handled using the techniques from \Cref{section:nfas-to-pqs}.

Another direction for future work is to relax the stratification condition in \Cref{definition:stratified-DL automaton} to capture a broader class of \HornALCHI queries expressible in GQL.
Indeed, our approach is not complete and fails to characterise queries such as $\langle \{\FP{A}(x) \wedge \FP{B}(x) \to \FP{C}(x), \FP{R}(x, y) \wedge \FP{A}(x) \to \FP{A}(x), \FP{S}(x, y) \wedge \FP{C}(x) \to \FP{B}(x)\}, \FP{C}(c)\rangle$ because this query is not equivalent to a PQ.
Nevertheless, it can be translated into an equivalent nested regular path query, which can in turn be expressed in GQL \cite[Theorem~11]{pods23:uc2rpq-to-gql}.
Considering nested regular path queries as a target language for our rewritings may therefore lead to a more general procedure.

Before diving into a deeper theoretical study, an empirical evaluation of our approach may be useful.
Namely, one could consider a collection of real-world \HornALCHI rule sets \cite{mowl-corpus} and determine which  are captured by our method, that is, determine how many correspond to a stratified DL automaton (cf. \Cref{definition:reduction-dl-query-to-DL automaton,definition:stratified-DL automaton}).
This would help assess \FirstItem whether our translation procedure is sufficiently general to be useful in practice, and \SecondItem whether simple modifications could be introduced to capture existing queries outside its current scope.

Finally, another relevant direction for future work is to characterize those \HornALCHI queries that do not admit an equivalent GQL formulation.
In particular, one could aim to obtain results analogous to those in \cite{DBLP:journals/ai/LutzS22}, which provide a decidable and complete characterization of $\mathcal{EL}$ queries that admit rewritings into first-order logic and linear Datalog.
Our translation into DL automata presented in \Cref{section:normalised-horn-alc-to-automata} may prove useful in this regard, as it could enable the transfer of techniques from automata theory, a well-developed area of theoretical computer science, for establishing inexpressibility results.

\paragraph*{Acknowledgements.} This research was funded by the Agence Nationale de la Recherche (ANR) under the Expand project (ANR-25-CE23-1215) and the France 2030 project (ANR-23-IACL-0008).



\bibliographystyle{splncs04}
\bibliography{bibliography}

\appendix

\section{Proofs of \Cref{section:horn-alc-to-normalised-horn-alc}}
\label{section:proofs-horn-alc-to-normalised-horn-alc}

In this section, we prove the claims leading to \Cref{theorem:normalization-is-correct}.
Let $\langle \RS, \Fact \rangle$ be a \HornALCHI query, and $\RS'$ the rule set obtained as described in \Cref{section:horn-alc-to-normalised-horn-alc}.

We first verify that $\langle \RS, \Fact \rangle$ and $\langle \RS', \Fact \rangle$ are equivalent.
Notice that axioms from $\RS'$ are consequences of $\RS$, and thus it is clear that $\QueryMapping(\RS', \Fact) \subseteq \QueryMapping(\RS, \Fact)$.
It remains to prove the converse, that is the following lemma.
\begin{lemma}
    $\QueryMapping(\RS', \Fact) \supseteq \QueryMapping(\RS, \Fact)$.
\end{lemma}
\begin{proof}
Let $\FS$ be a fact set such that $\langle \RS, \FS \rangle \models \alpha$, we need to prove that $\langle \RS', \FS \rangle \models \alpha$.
Let $\Imc$ be an FO-model of $\langle \RS', \FS \rangle$, we denote $\Delta^\Imc$ its domain and $\cdot^\Imc$ its interpretation function.
We construct an FO-structure $\Jmc$ that extends $\Imc$ and will be a model of $\langle \RS, \FS \rangle$.
The domain of $\Jmc$ is also $\Delta^\Imc$, it interprets every class as in $\Imc$, and every role $\FP{R} \in \Roles$ as follows:
\[
\FP{R}^\Jmc := \bigcup_{\substack{\FP{S} \in \RoleExpressions \\
\text{s.t.\ } \FP{S} \roleclosure \FP{R}}} \FP{S}^\Imc.
\]
We further refer to the property that, by definition, $\Jmc$ and $\Imc$ agree on the interpretation of classes as $(\star)$.
We now verify that $\Jmc$ is a model of $\langle \RS, \FS \rangle$, which will conclude the proof since it guarantees $\Jmc \models \alpha$ using the assumption that $\langle \RS, \FS \rangle \models \alpha$ and consequently $\Imc \models \alpha$ by \propconcepts.
It is clear that $\Jmc \models \FS$, it thus remains to verify $\Jmc \models \RS$.
We examine each possible form of axiom in $\RS$:
\begin{enumerate}
    \item[\eqref{rule:conjunction}] Immediate by \propconcepts.
    \item[\eqref{rule:lhs-existential}] Assume $\FP{R}(x, y) \wedge \FP{A}(y) \to \FP{B}(x) \in \RS$, $(d, e) \in \FP{R}^\Jmc$ and $e \in \FP{A}^\Jmc$.
    By definition of $\RS'$, we have $\FP{R}(x, y) \wedge \FP{A}(y) \to \FP{B}(x) \in \RS'$ and, by \propconcepts, we have $e \in \FP{A}^\Imc$.
    By definition of $\FP{R}^\Jmc$, we have $(d, e) \in \FP{S}^\Imc$ for some $\FP{S} \roleclosure \FP{R}$.
    Now, either $\FP{S} \in \Roles$, and then we have $\FP{S}(x, y) \wedge \FP{A}(y) \to \FP{B}(x) \in \RS'$ by definition of $\RS'$, or $\FP{S} \in \InvRoles$ and we have $\FP{A}(x) \wedge \FP{S}(x, y) \to \FP{B}(y) \in \RS'$.
    Either way, $\Imc \models \RS'$ yields $d \in \FP{B}^\Imc$, and we conclude using \propconcepts.
    \item[\eqref{rule:universal}] 
    Assume $\FP{A}(x) \wedge \FP{R}(x, y) \to \FP{B}(y) \in \RS$, $(d, e) \in \FP{R}^\Jmc$ and $d \in \FP{A}^\Jmc$.
    By definition of $\RS'$, we have $\FP{A}(x) \wedge \FP{R}(x, y) \to \FP{B}(y) \in \RS'$ and, by \propconcepts, we have $d \in \FP{A}^\Imc$.
    By definition of $\FP{R}^\Jmc$, we have $(d, e) \in \FP{S}^\Imc$ for some $\FP{S} \roleclosure \FP{R}$.
    Now, either $\FP{S} \in \Roles$, and then we have $\FP{A}(x) \wedge \FP{S}(x, y)  \to \FP{B}(y) \in \RS'$ by definition of $\RS'$, or $\FP{S} \in \InvRoles$ and we have $\FP{S}(x, y) \wedge \FP{A}(y) \to \FP{B}(x) \in \RS'$.
    Either way, $\Imc \models \RS'$ yields $e \in \FP{B}^\Imc$, and we conclude using \propconcepts.
    \item[\eqref{rule:top}] Immediate by \propconcepts.
    \item[\eqref{rule:rhs-existential}]
    Assume $\FP{A}(x) \to \exists y . \FP{R}(x, y) \wedge \FP{B}(y) \in \RS$ and $d \in \FP{A}^\Jmc$.
    By definition of $\RS'$, we have $\FP{A}(x) \to \exists y . \FP{R}(x, y) \wedge \FP{B}(y) \in \RS'$ and, by \propconcepts, we have $d \in \FP{A}^\Imc$.
    Since $\Imc \models \RS'$, there exists $e \in \FP{B}^\Imc$ such that $(d, e) \in \FP{R}^\Imc$.
    Using that $\roleclosure$ is reflexive, the definition of $\FP{R}^\Jmc$ yields $(d, e) \in \FP{R}^\Jmc$.
    Using \propconcepts, we obtain $e \in \FP{B}^\Jmc$ and we are done.
    \item[\eqref{rule:subrole}]
    Immediate by construction of $\Jmc$ and $\roleclosure$ being transitive.
    \item[\eqref{rule:inverse}]
    Immediate by construction of $\Jmc$ and $\roleclosure$ being transitive.
\end{enumerate}
\end{proof}

We now explain how the Datalog translation proposed in \cite{normal-form-horn-shiq} can be applied to $\RS'$.
Let us highlight that $\RS'$ complies with the normal form demanded in \cite{normal-form-horn-shiq}: axioms of $\RS'$ with forms \eqref{rule:conjunction} or \eqref{rule:top} are exactly in their form (F1);
those with forms \eqref{rule:lhs-existential} or \eqref{rule:universal} are exactly in their form (F2) (since the latter supports inverse roles!); and those with form \eqref{rule:rhs-existential} are exactly in their form (F3).
In particular, $\RS'$ does not contain any axiom in their form (F4).
In turn, it follows that, when applying the inference rules from Table~2 in \cite{normal-form-horn-shiq}, their Rules~{$\mathbf{R}_\leq$} and $\mathbf{R}^-_\leq$ are never applied.
The same holds for their Rule~$\mathbf{R}^r_\sqsubseteq$ since $\RS'$ does not contain role inclusions, and Rule~$\mathbf{R}_\bot$ since we restricted our attention to positive \HornALCHI rules.
These observations further guarantee that, applying the final translation in Table~3 \cite{normal-form-horn-shiq}, only the first two completion rules are ever applied.
Therefore, we obtain a set $\RS''$ of Datalog rules whose rules are of form 
$B(y) \leftarrow A(x), r(x, y)$ and $B(x) \leftarrow A_1(x), \dots , A_n(x)$, where $r$ is possibly an inverse role.
Now, $B(x) \leftarrow A_1(x), \dots , A_n(x)$ rules are either in forms \eqref{rule:conjunction} or \eqref{rule:top} in our formalism, while $B(y) \leftarrow A(x), r(x, y)$ rules are either in form \eqref{rule:universal} if $r \in \Roles$, or in form \eqref{rule:lhs-existential} if $r \in \InvRoles$.
Altogether, we obtained a set $\RS''$ of $\HornALCDat$ rules such that, for every class fact $\Fact$, $\langle \RS'', \Fact \rangle$ is equivalent to $\langle \RS', \Fact \rangle$ \cite[Theorem~3]{normal-form-horn-shiq}, thus equivalent to $\langle \RS, \Fact \rangle$.
This concludes the proof of Theorem~\ref{theorem:normalization-is-correct}.

\section{Proofs of \Cref{section:normalised-horn-alc-to-automata}}
\label{section:proofs-normalised-horn-alc-to-automata}

In this section, we prove \Cref{theorem:correspondence-normalised-horn-alc-automata}, which establishes the equivalence between a \HornALCDat query \Query and the corresponding automaton $\HornALCToDLAutomaton(\Query)$.
More precisely, we prove \Cref{lemma-hornalc-aut,lemma-aut-hornalcdat}, from which the theorem follows directly as a corollary.

\begin{lemma}
\label{lemma-aut-hornalcdat}
A \HornALCDat query \Query is contained in the DL automaton $\HornALCToDLAutomaton(\Query)$.
\end{lemma}

\begin{proof}
We fix a \HornALCDat query $\Query = \langle \RS, \FP{A}(a) \rangle$, the DL automaton $\HornALCToDLAutomaton(\Query) = \langle a, Q, q_\FP{A}, \TF \rangle$, and a fact set $\FS$.
To prove the lemma, we assume that $\FS \in \QueryMapping(\Query)$ and show that this implies $\FS \in \QueryMapping(\HornALCToDLAutomaton(\Query))$.
More precisely, we establish that for any given fact $\FP{B}(b) \in \FF{Ch}(\RS, \FS)$, there is an accepting run of~$\HornALCToDLAutomaton(\RS, \FP{B}(b))$ over \FS from $\langle q_\FP{B}, b \rangle$.
Applying this claim to the fact $\FP{A}(a)$, we obtain an accepting run of $\HornALCToDLAutomaton(\Query)$ over \FS from $\langle q_\FP{A}, a \rangle$, and therefore $\FS \in \QueryMapping(\HornALCToDLAutomaton(\Query))$.
We prove the claim by induction on the depth of $\FP{B}(b)$ with respect to $\langle \RS, \FS \rangle$.
Recall also that $\FP{B}(b) \in \FF{Ch}(\RS, \FS)$ if and only if $\langle \RS, \FS \rangle \models \FP{B}(b)$ by \Cref{propoosition:chase-correctness}.

For the base case, we assume that $\FP{B}(b)$ has depth~1 with respect to $\langle \RS, \FS \rangle$, which implies that $\FP{B}(b) \in \FS$.
Consider now the directed labelled tree $\langle \{r, r'\}, \{r \to r'\}, L \rangle$, where $L$ is the labelling function mapping $r$ to $\langle q_{\FP{B}}, b \rangle$ and $r'$ to $\langle \AcceptingState, b \rangle$.
This tree is an accepting run of $\HornALCToDLAutomaton(\Query)$ over \FS from $\langle q_\FP{B}, b \rangle$. 
Observe that it has a unique leaf labelled with the accepting state.
Moreover, the edge $r \to r'$ is valid, since $\FP{B}(b) \in \FS$ and $\Transition{q_{\FP{B}}}{\FP{B}}{\{\AcceptingState\}} \in \TF$.

For the induction step, we assume that $\FP{B}(b)$ has depth $k \geq 2$ with respect to $\langle \RS, \FS \rangle$.
Then at least one of the following cases applies:
\begin{itemize}
%
\item  There is a rule of the form $\FP{R}(x,y) \land \FP{C}(y) \to \FP{B}(x) \in \RS$ and a constant $c$ such that $\FP{R}(b,c), \FP{C}(c) \in \FF{Ch}_{k-1}(\RS, \FS)$. 
By induction hypothesis, there is an accepting run $\langle N, E, L \rangle$ of $\HornALCToDLAutomaton(\RS, \FP{C}(c))$ over~$\FS$ from $\langle q_\FP{C}, c \rangle$.
Consider the directed labelled tree $\langle N', E', L' \rangle$ where $N' = N \cup \{r\}$ for some fresh node~$r$, $E' = E \cup \{r \to r'\}$ where $r'$ is the root node of $\langle N, E, L \rangle$, and $L'$ is the function that extends $L$ and maps $r$ to $\langle q_\FP{B}, b \rangle$.
We argue that~$\langle N', E', L' \rangle$ is an accepting run for $\HornALCToDLAutomaton(\RS, \FP{B}(b))$ over $\FS$ from $\langle q_\FP{B}, b \rangle$:
\begin{itemize}
\item Since $\FP{R}(x,y) \land \FP{C}(y) \to \FP{B}(x) \in  \RS$, we have that $\Transition{q_{\FP{B}}}{\FP{R}}{\{q_{\FP{C}}\}} \in \TF$.
Since no rule allows the derivation of binary facts, we have that $\FP{R}(b,c) \in \FS$, which implies that the edge $r \to r'$ is valid as per \Cref{definition-run-automaton}.
Therefore $\langle N', E', L' \rangle$ is a run, as all other edges in this tree satisfy the same definition by induction hypothesis.
\item The leaves in $\langle N', E', L' \rangle$ are exactly the leaves of $\langle N,E,L \rangle$.
Hence, all of the leaves in $\langle N',E',L'\rangle$ feature the accepting state.
\end{itemize}
\item There is a rule of the form $\FP{C}(y) \wedge \FP{R}(x,y) \to \FP{B}(x) \in \RS$ and a constant $c$ such that $\FP{C}(c), \FP{R}(c, b) \in \FF{Ch}_{k-1}(\RS, \FS)$. 
This case is analogous to the previous one.
\item There is a rule of the form $\FP{C}_1(x) \wedge \ldots \wedge \FP{C}_n(x) \to \FP{B}(x) \in \RS$ with $\FP{C}_1(b), \ldots, \FP{C}_n(b) \in \FF{Ch}_{k-1}(\RS, \FS)$.
By induction hypothesis, there is an accepting run $\langle N_i, E_i, L_i \rangle$ of $\HornALCToDLAutomaton(\RS, \FP{C}_i(b))$ over~$\FS$ from $\langle q_{\FP{C}_i}, b \rangle$ for every $1 \leq i \leq n$.
Consider the directed labeled tree $\langle N', E', L' \rangle$ where $N' = N_1 \cup \ldots \cup N_n \cup \{r\}$ with $r$ a fresh node, $E' = E_1 \cup \ldots \cup E_n \cup \{r \to r_1 , \ldots, r \to r_n\}$ with $r_i$ the root node of $\langle N_i, E_i, L_i \rangle$ for every $1 \leq i \leq n$, and $L'$ is the function that extends $L_i \cup \ldots \cup L_n$ and maps $r$ to~$\langle q_\FP{B}, b \rangle$.
We argue that~$\langle N', E', L' \rangle$ is an accepting run for $\HornALCToDLAutomaton(\RS, \FP{B}(b))$ over $\FS$ from $\langle q_\FP{B}, b \rangle$:
\begin{itemize}
\item As $\FP{C}_1(x) \wedge \ldots \wedge \FP{C}_n(x) \to \FP{B}(x) \in \RS$, we have that $\Transition{q_{\FP{B}}}{\EmptyTransition}{\{q_{\FP{C}_1}, \ldots, q_{\FP{C}_n}\}} \in \delta$.
This implies that the edges $r \to  r_1$ through $r \to  r_n$ in $E'$ are valid as per \Cref{definition-run-automaton}.
Therefore $\langle N', E', L' \rangle$ is a run, as all other edges in this tree satisfy the same definition by induction hypothesis.
\item The leaves in $\langle N', E', L' \rangle$ are exactly the union of the leaves of the $\langle N_i,E_i,L_i \rangle$.
Hence, all of the leaves in $\langle N',E',L'\rangle$ feature the accepting state.
\end{itemize}
\item There is a rule of the form $\to \FP{B}(x) \in \RS$.
Consider the directed labelled tree $\langle {r, \ell}, {r \to \ell}, L \rangle$, where $L$ maps $r$ to $\langle q_{\FP{B}}, b \rangle$ and $\ell$ to $\langle \AcceptingState, b \rangle$.
This tree is an accepting run of $\HornALCToDLAutomaton(\RS, \FP{B}(b))$ over \FS starting from $\langle q_\FP{B}, b \rangle$.
Note that $\Transition{q_\FP{B}}{\EmptyTransition}{\{\AcceptingState\}} \in \TF$ since $\to \FP{B}(x) \in \RS$.
\end{itemize}
In either case, there is an accepting run of $\HornALCToDLAutomaton(\RS, \FP{B}(b))$ over $\FS$ from $\langle q_\FP{B}, b \rangle$.
\end{proof}

\begin{lemma}
\label{lemma-hornalc-aut}
A \HornALCDat query \Query contains the DL automaton $\HornALCToDLAutomaton(\Query)$.
\end{lemma}

\begin{proof}
We fix a $\HornALCDat$ query $\Query = \langle \RS, \FP{A}(a) \rangle$ and a fact set $\FS$.
To prove the lemma, we assume that $\FS \in \FF{QM}(\HornALCToDLAutomaton(\Query))$ and show that $\FS \in \FF{QM}(\Query)$.
More precisely, we argue that, given some fact $\FP{B}(b)$, if there is an accepting run of $\HornALCToDLAutomaton(\RS, \FP{A}(a))$ over $\FS$ from $\langle q_\FP{B}, b \rangle$, then $\FP{B}(b) \in \FF{Ch}(\RS, \FS)$.
This implies in particular that $\FP{A}(a) \in \FF{Ch}(\RS, \FS)$, and therefore that $\FS \in \FF{QM}(\Query)$.
The proof proceeds by induction on the height of the accepting run of the fact $\FP{B}(b)$.
Recall that a fact is entailed by $\langle \RS, \FS \rangle$ if and only if it is in $\FF{Ch}(\RS, \FS)$ by \Cref{propoosition:chase-correctness}.

Given some fact $\FP{B}(b)$, there is no accepting run of $\HornALCToDLAutomaton(\RS, \FP{A}(a))$ over $\FS$ from~$\langle q_\FP{B}, b \rangle$ of height $1$.
Indeed, the unique node of such a run would have to feature both $q_\FP{B}$ and \AcceptingState in its label, which is impossible since these are distinct states.
Hence, the base case of our induction starts with runs of height $2$.

Regarding the base case, we assume that there is an accepting run of $\HornALCToDLAutomaton(\RS, \FP{A}(a))$ over $\FS$ from $\langle q_\FP{B}, b \rangle$ of height 2.
In this case, the root node of this run is labelled~$\langle q_{\FP{B}}, b \rangle$ and all of the labels of its children feature the accepting state.
Therefore, there are only two possible cases:
\begin{itemize}
\item We have that $\FP{B}(b) \in \FS$.
Note that $\Transition{q_\FP{B}}{\FP{B}}{\{\AcceptingState\}}$ is a transition of $\HornALCToDLAutomaton(\Query)$.
\item We have that $\Transition{q_\FP{B}}{\EmptyTransition}{\{\AcceptingState\}}$ is a transition of $\HornALCToDLAutomaton(\Query)$.
Note that, in this case, we have that $\to \FP{B}(x) \in \RS$.
\end{itemize}
In either case, we have that $\FP{B}(b) \in \FF{Ch}(\RS, \FS)$.
Recall that $b$ occurs in $\FS$ since there is an accepting run of $\HornALCToDLAutomaton(\RS, \FP{A}(a))$ over $\FS$ from $\langle q_\FP{B}, b \rangle$.

Regarding the induction step, let us consider an accepting run $\langle N,E,L \rangle$ for $\HornALCToDLAutomaton(\Query)$ over $\FS$ from $\langle q_\FP{B}, b \rangle$ of height $k >2$.
We consider two different cases depending if the root node $r$ of this run has one child or more.
\begin{itemize}
\item Assume that the root $r$ has a unique child $r'$, which does not feature \AcceptingState in its label by \Cref{remark:DL automaton-properties}.
Then, by \Cref{definition-run-automaton} and \Cref{remark:DL automaton-properties}, one of the following cases must hold:
\begin{itemize}
\item There is some class $\FP{C}$ such that $r'$ is labelled with $\langle q_{\FP{C}}, b \rangle$, $\Transition{q_\FP{B}}{\EmptyTransition}{{\{q_\FP{C}\}}} \in \HornALCToDLAutomaton(\Query)$, and $\FP{C}(b) \in \FS$.
Now consider the tree $\langle N', E', L' \rangle$ obtained by removing $r$ from $\langle N, E, L \rangle$, which is of height $k-1$ and is an accepting run for $\HornALCToDLAutomaton(\Query)$ from $\langle q_\FP{C}, b \rangle$ over \FS.
Hence, $\FP{C}(b) \in \FF{Ch}(\RS, \FS)$ by induction hypothesis.
Therefore, $\FP{B}(b) \in \FF{Ch}(\RS, \FS)$ since $\FP{C}(x) \to \FP{B}(x) \in \RS$.
\item There are a class $\FP{C}$ and a constant $c$ such that $r'$ is labelled with $\langle q_{\FP{C}}, c \rangle$, $\Transition{q_\FP{B}}{\FP{R}}{{q_\FP{C}}} \in \HornALCToDLAutomaton(\Query)$, and $\FP{R}(b,c) \in \FS$.
Now consider the tree $\langle N', E', L' \rangle$ obtained by removing $r$ from $\langle N, E, L \rangle$, which is of height $k-1$ and is accepting for $\HornALCToDLAutomaton(\RS, \FP{C}(c))$ over \FS from $\langle q_\FP{C}, c \rangle$.
Hence, $\FP{C}(c) \in \FF{Ch}(\RS, \FS)$ by the induction hypothesis.
Therefore, $\FP{B}(b) \in \FF{Ch}(\RS, \FS)$ since $\FP{R}(x, y) \wedge \FP{C}(y) \to \FP{B}(x) \in \RS$.
\item There are a class $\FP{C}$ and a constant $c$ such that $r'$ is labelled with $\langle q_{\FP{C}}, c \rangle$, $\Transition{q_\FP{B}}{\FP{R}^-}{{q_\FP{C}}} \in \HornALCToDLAutomaton(\RS, \FP{C}(c))$, and $\FP{R}(c, b) \in \FS$.
This case is analogous to the previous one.
\end{itemize}

\item Consider the root $r$ and its multiple children $r_1, \ldots, r_n$.
By \Cref{definition-run-automaton}, there are some $\FP{C}_1, \ldots, \FP{C}_n \in \Classes$ such that $r_1$ is labelled with $\langle q_{\FP{C}_1}, b \rangle$, $r_2$ with $\langle q_{\FP{C}_2}, b \rangle$, and so on.
Moreover, $\Transition{q_\FP{B}}{\EmptyTransition}{\{q_{\FP{C}_1}, \ldots,  q_{\FP{C}_n}\}}$ in $\HornALCToDLAutomaton(\Query)$.
Now, for every $1 \leq i \leq n$, consider the subtree $\langle N_i, E_i, L_i \rangle$ of $\langle N, E, L \rangle$ rooted at $r_i$, which is of height at most $k-1$ and is an accepting run of $\HornALCToDLAutomaton(\Query)$ over $\FS$ from $\langle q_{\FP{C}_i}, b \rangle$.
From the existence of this trees, we conclude that $\FP{C}_1(b), \ldots, \FP{C}_n(b) \in \FF{Ch}(\RS, \FS)$ by induction hypothesis.
Therefore, $\FP{B}(b) \in \FF{Ch}(\RS, \FS)$ since $\FP{C}_1(x) \wedge \ldots \wedge \FP{C}_n(x)  \to \FP{B}(x) \in \RS$.
\end{itemize}
In either case above, we have $\FP{B}(b) \in \FF{Ch}(\RS, \FS)$.
\end{proof}


\section{Proofs of \Cref{section:automata-to-acyclic-pqs}}


In this section, we provide complete proofs of the results from \Cref{section:automata-to-acyclic-pqs}.
More precisely, we prove Lemmas~\ref{lemma-transition-crpq}, \ref{lemma-crpq-dlaut}, and \ref{lemma:dl-aut-in-strat-query}.
Note that \Cref{theorem:correspondence-automata-crpq} follows directly as a corollary of the latter two lemmas.

\lemmatransitioncrpq*

\begin{proof}
We first prove the forward implication of the equivalence stated in the lemma.
Assume that there is an accepting run of $\NFA$ on $\FS$ starting from $\langle q, s \rangle$, and let $n_1, \ldots, n_k$ be the sequence of nodes occurring, in order, along the unique branch of this run.
Note that $k \geq 2$ since $q$ is different from the accepting state.
Then we are in one of the following cases:
\begin{itemize}
\item If $k = 2$, then either $\Transition{q}{\EmptyTransition}{\{\AcceptingState\}} \in \NFA$ or, otherwise, $\Transition{q}{\FP{C}}{\{\AcceptingState\}} \in \NFA$ for some $\FP{C} \in \Classes$.
In the first case, the disjunct $\top$ is in $\NFAToPostiveQuery(\NFA, q, s)$, and that disjunct is entailed by $\FS$. In the second case, the disjunct $\FP{C}(s)$ is in $\NFAToPostiveQuery(\NFA, q, s)$, and that disjunct is entailed by $\FS$ since $\FP{C}(s) \in \FS$.
Note that this fact is in \FS because the run is accepting and $\Transition{q}{\EmptyTransition}{\{\AcceptingState\}} \notin \NFA$.
\item If $k \geq 3$, then consider the node $n_{k-1}$ and its label $\langle p, t \rangle$.
Since the NFA can transition from $\langle q, s \rangle$ to $\langle p, t \rangle$ in $\FS$, we have that $\FS$ entails $\RE^\NFA_{q \triangleright p}(s,t)$ from \Cref{lemma-nfa-transition}.
Now, one of the following cases must hold:
\begin{itemize}
\item If $\Transition{p}{\EmptyTransition}{\{\AcceptingState\}} \in \NFA$, then the disjunct $\exists x . \RE^\NFA_{q \triangleright p}(s,x)$ is in $\NFAToPostiveQuery(\NFA, q,s)$ and \FS entails this disjunct, since this fact set entails $\RE^\NFA_{q \triangleright p}(s,t)$.
\item If not, then there is some $\FP{C} \in \Classes$ such that $\Transition{p}{\FP{C}}{{\AcceptingState}} \in \NFA$, and, since the run is accepting, $\FP{C}(t) \in \FS$.
In this case, the disjunct $\exists x . \RE^\NFA_{q \triangleright p}(s,x) \wedge \FP{C}(x)$ occurs in $\NFAToPostiveQuery(\NFA, q, s)$ and is entailed by $\FS$.
This is indeed the case since $\FS$ entails both $\RE^\NFA_{q \triangleright p}(s,t)$ and $\FP{C}(t)$.
\end{itemize}
\end{itemize}
In either case, we have that $\FS$ satisfies $\NFAToPostiveQuery(\NFA, q,s)$.

For the reverse direction of the equivalence, assume that $\FS$ satisfies $\NFAToPostiveQuery(\NFA, q, s)$.
Then $\FS$ satisfies at least one disjunct \Formula of this PQ.
Depending on the type of disjunct satisfied by $\FS$, one of the following cases applies:
    \begin{itemize}
        \item If $\Formula = \top$, then there is a transition $\Transition{q}{\EmptyTransition}{\{\AcceptingState\}} \in \NFA$.
        In that case, the run from $\langle q,s \rangle$ to $\langle \AcceptingState, s \rangle$ is accepting for $\FS$ in $\NFA$.
        \item If $\Formula = \FP{C}(s)$ for some $\FP{C} \in \Classes$, then we have $\Transition{q}{\FP{C}}{\{\AcceptingState\}} \in \NFA$. In that case, the run $\FS$ from $\langle q,s \rangle$ to $\langle \AcceptingState, s \rangle$ is accepting since $\FP{C}(s) \in \FS$.
        \item Assume that \Formula is of the form $ \exists w . \RegularExpressionNFAStateToState{q}{p}{\NFA}(s, w) \wedge \FP{C}(w)$. 
        Then, there is some constant $t$ such that \FS entails $\RegularExpressionNFAStateToState{q}{p}{\NFA}(s, t)$ and $\FP{C}(t)$,
        and $\Transition{p}{\FP{C}}{\{\AcceptingState\}} \in \NFA$.
        From \Cref{lemma-nfa-transition}, we conclude that \NFA can transition from $\langle q, s \rangle$ to $\langle p,t \rangle$ in~\FS, that is, that there is a run of \NFA over \FS that starts on a node labelled with $\langle q, s \rangle$ and ends on a node labelled $\langle p,t \rangle$.
        If we append one more node at the end of this run labelled with $\langle \AcceptingState, t \rangle$, we obtain an accepting run of \NFA starting $\langle q, s \rangle$ over \FS; recall that $\Transition{p}{\FP{C}}{\{\AcceptingState\}} \in \NFA$ and that $\FP{C}(t) \in \FS$.
        \item Assume that \Formula is of the form $ \exists w . \RegularExpressionNFAStateToState{q}{p}{\NFA}(s, w)$. 
        Then, there is some constant $t$ such that \FS entails $\RegularExpressionNFAStateToState{q}{p}{\NFA}(s, t)$,
        and $\Transition{p}{\EmptyTransition}{\{\AcceptingState\}} \in \NFA$.
        From \Cref{lemma-nfa-transition}, we conclude that \NFA can transition from $\langle q, s \rangle$ to $\langle p,t \rangle$ in~\FS, that is, that there is a run of \NFA over \FS that starts on a node labelled with $\langle q, s \rangle$ and ends on a node labelled $\langle p,t \rangle$.
        If we append one more node at the end of this run labelled with $\langle \AcceptingState, t \rangle$, we obtain an accepting run of \NFA starting $\langle q, s \rangle$ over \FS; recall that $\Transition{p}{\EmptyTransition}{\{\AcceptingState\}} \in \NFA$.
        
    \end{itemize}

    In any case, we have an accepting run for $\NFA$ from $\langle q, s \rangle$ over $\FS$.
\end{proof}

We then show a preliminary lemma allowing us to substitute an automaton $\DLAut$ to its extended version: for the following, we will assume that $\DLAut$ is in extended form.

\begin{lemma}\label{equivalence-epsilon-automaton}
Consider an automaton $\DLAut= \langle s, Q, \InitialState, \delta \rangle$ and its extended version $\DLAut_\EmptyTransition =\langle s, Q, \InitialState,  \delta \cup \delta_\EmptyTransition^* \rangle$.
For every fact set \FS, state $p \in Q$, and constant $c$, there is an accepting run of \DLAut over \FS from $\langle c, p \rangle$ if and only if there is an accepting run of $\DLAut_\EmptyTransition$ over \FS from $\langle c, p \rangle$.
Therefore, $\DLAut$ and $\DLAut_\EmptyTransition$ are equivalent.
\end{lemma}

\begin{proof}
    The inclusion of $\DLAut$ in $\DLAut_\EmptyTransition$ is trivial. For the other direction,
     we show that every $\EmptyTransition$- transition of $\DLAut_\EmptyTransition$ can be simulated by a finite sequence of $\EmptyTransition$-transitions in $\delta$: from this, it follows that if there is an accepting run on $\DLAut_\EmptyTransition$ for some fact set $\FS$, starting node $\langle p,c \rangle$ then there is an accepting run on $\DLAut$ for the same fact set and starting node.

First note that we can express $\delta_\EmptyTransition^*$ as the (finite) limit of some $\delta_0, \ldots, \delta_n, \ldots$ defined such that $\delta_0= \{ \Transition{q}{\EmptyTransition}{q} \mid q \in Q\}$ and for every $i \geq 0$, $\delta_{i+1}= \delta_i \cup \{ \Transition{q}{\EmptyTransition}{(P \setminus \{p\}) \cup P'} \mid q \in Q\}$ where $p$ is a state, and $P$ and $P'$ are sets of states such that $\Transition{q}{\EmptyTransition}{P} \in \delta_i$, and $\Transition{p}{\EmptyTransition}{P'} \in \delta$, and $p \in P$. 
We then show, by induction on $i$, that every $\delta_i$ can be simulated by a finite sequence of transitions of $\delta$.
        For the base case, if the transition is of form $\Transition{q}{\EmptyTransition}{\{q\}}$, then the transition can be simulated by an empty sequence of transitions: therefore $\delta_0$ can be simulated by a finite sequence of transitions of $\delta$.
        
        For the induction step, we assume that $\delta_i \subseteq \delta_{\EmptyTransition}^*$, and consider $\delta_{i+1}= \delta_i \cup \{\Transition{q}{\EmptyTransition}{Q'}\}$. Then there are some $P$, $p$, $P'$ so that $\Transition{q}{\EmptyTransition}{P} \in \delta_i$, and $p \in P$, and $\Transition{p}{\EmptyTransition}{P'} \in \delta$, and $Q' = (P \setminus \{p\}) \cup P'$. 
        Since $\Transition{q}{\EmptyTransition}{P} \in \delta_i$, it can be simulated by a finite sequence $t_1, \ldots, t_n$ of transitions in $\delta$. Then, the sequence of transitions $t_1, \ldots, t_n, \Transition{p}{\EmptyTransition}{P'}$ simulates $\Transition{q}{\EmptyTransition}{Q'}$: indeed, the transition $\Transition{p}{\EmptyTransition}{P'}$ can be added to the sequence since $p \in P$. Doing so removes $p$ and adds $P'$ to the states at the end of the sequence: therefore, the states at the end of the new sequence are $ (P \setminus \{p\}) \cup P'=Q'$. Therefore we can simulate $\delta_{i+1}$ with a finite sequence of transitions of $\delta$.
\end{proof}

    We then move on to the two main lemmas.
\begin{lemma}\label{lemma-crpq-dlaut}
A stratified DL automaton $\DLAut$ contains $\StratDLAToPostiveQuery(\DLAut)$.
\end{lemma}
\begin{proof}
We fix a stratified DL automaton $\DLAut = \langle s, Q, \InitialState, \TF \rangle$, its extended version~$\DLAut_\EmptyTransition =\langle s, Q, \InitialState,  \delta \cup \delta_\EmptyTransition^* \rangle$, the query $\StratToPQ(\DLAut)$, and a fact set $\FS$. 
To prove the lemma, we assume that $\FS$ entails $\StratToPQ(\DLAut)$ and show that $\DLAut_\EmptyTransition$ accepts $\FS$.
More precisely, we prove the following claim by induction on $k \geq 1$: if $\FS$ entails $\StratDLAToPostiveQuery(\DLAut, p, c)$ for some state $p$ of $\DLAut$ in stratum $k$ and some constant $c$, then there is an accepting run of $\DLAut_\EmptyTransition$ on $\FS$ starting from $\langle p, c \rangle$.
Applying this claim to $p = \InitialState$ and $c = s$, we obtain an accepting run of $\DLAut_\EmptyTransition$ on $\FS$ starting from $\langle \InitialState, s \rangle$, and hence $\DLAut_\EmptyTransition$ accepts~$\FS$.
Recall that $\DLAut$ and $\DLAut_\EmptyTransition$ are equivalent by \Cref{equivalence-epsilon-automaton} and therefore, the lemma follows.


For the base case, let $p$ be a state in stratum $1$ and $c$ a constant such that~$\FS$ satisfies $\StratDLAToPostiveQuery(\DLAut, p, c)$.
Since $p$ belongs to stratum $1$, no state succeeds $p$ with respect to \InHigherStrata{\DLAut}{}{}.
Therefore, by \Cref{definition:reduction-stratified-DL automaton-to-pq}, we have:
    $$\StratDLAToPostiveQuery(\DLAut, p, c)= \bigvee\nolimits_{\Transition{p}{\EmptyTransition}{P} \in \TF^*_\EmptyTransition} \bigwedge\nolimits_{q \in P} \NFAToPostiveQuery(\FF{NFA}(\DLAut,p), q, c)$$
    Since this PQ is entailed by \FS, there is some $\Transition{p}{\EmptyTransition}{\{q_1, \ldots, q_m\}} \in \TF^*_\EmptyTransition$ such that $\FS$ satisfies $\bigwedge\nolimits_{1 \leq i \leq m} \NFAToPostiveQuery(\FF{NFA}(\DLAut,p), q_i, c)$.
    Since this formula is a conjunction, \FS satisfies $\NFAToPostiveQuery(\FF{NFA}(\DLAut, p), q_i, c)$ for every $1 \leq i \leq m$.
    Hence, there is an accepting run for $\DLAut$ over $\FS$ from every $\langle q_i, c \rangle$  by \Cref{lemma-transition-crpq}.
    Therefore, for every $1 \leq i \leq m$, there is an accepting run $\langle N_i, E_i, L_i \rangle$ for $\DLAut_\EmptyTransition$ over $\FS$ from $\langle q_i, c \rangle$ by \Cref{equivalence-epsilon-automaton}.
    Without loss of generality, we assume that the sets $N_1, \ldots, N_m$ are mutually disjoint.
    Now, consider the tree $\langle N, E, L\rangle$ defined as follows:
    \begin{itemize}
    \item Let $N = N_1 \cup \ldots \cup N_m \cup \{r\}$ with $r$ a fresh node that is the root of this tree.
    \item Let $E = E_1 \cup \ldots \cup E_m \cup \{r \to r_1, \ldots, r \to r_m\}$ with $r_1$ the root of $\langle N_1, E_1, L_1 \rangle$, $r_2$ the root of $\langle N_2, E_2, L_2 \rangle$, and so on.
    \item Let $L$ be the mapping extending $L_i$ for every $1 \leq i \leq m$ and maps $r$ to $\langle p, c \rangle$.
    \end{itemize}
    We argue that $\langle N, E, L\rangle$ is an accepting run for $\DLAut_\EmptyTransition$ over $\FS$ from $\langle p, c \rangle$.
    This follows because $\Transition{p}{\EmptyTransition}{\{q_1, \ldots, q_m\}} \in \TF^*_\EmptyTransition$ and hence, the edges $\{r \to r_1, \ldots, r \to r_m\}$ in this tree satisfy \Cref{definition-run-automaton}.
    Moreover, for every $1 \leq i \leq m$, we have that $\langle N_i, E_i, L_i \rangle$ is an accepting run over $\FS$ from $\langle q_i, c \rangle$.
    Hence, all of the edges in $E_1 \cup \ldots \cup E_m$ also satisfy this definition, and all of the leaves in $\langle N, E, L\rangle$ feature the accepting state.

For the induction step, we consider a state $p$ of stratum $k \geq 2$ and a constant~$c$ such that $\FS$ satisfies $\StratDLAToPostiveQuery(\DLAut, p, c)$.
Since \FS satisfies $\StratDLAToPostiveQuery(\DLAut, p, c)$, it satisfies some disjunct, there is some $\Transition{p}{\EmptyTransition}{\{q_1, \ldots, q_m\}} \in \TF_\EmptyTransition^*$ such that \FS satisfies the following:
    \begin{align*} & \bigwedge\nolimits_{1 \leq i \leq m}  \Big( \NFAToPostiveQuery(\FF{NFA}(\DLAut,p), q_i, c)  \\ &
    \bigvee_{\substack{\forall p', p_1,\ldots, p_n \in Q \textit{ such that } \TRTransRel{\NFA}{q_i}{p'}, \\ \Transition{p'}{\EmptyTransition}{\{p_1,\ldots,  p_n\}},~\InHigherStrata{\DLAut}{q_i}{p_1}, \ldots, \InHigherStrata{\DLAut}{q_i}{p_n}}} \big(\exists w. \RegularExpressionNFAStateToState{q_i}{p'}{\NFA}(c, w) \bigwedge_{1 \leq j \leq n} \StratDLAToPostiveQuery(\DLAut, p_j, w)  \big) \Big)\end{align*} 
    In the remainder of our argument, we define an accepting run $\langle N, E, L \rangle$ for $\DLAut_\EmptyTransition$ over \FS from $\langle p, c \rangle$.
    The root node $r \in N$ of this run is mapped by $L$ to $\langle p, c \rangle$.
    Moreover, there are some distinct nodes $r_1, \ldots, r_m \in N$ such that, for every $1 \leq i \leq m$, we have that $r \to r_i \in E$ and $L(r_i) = \langle q_i, c \rangle$.
    Now, for every $1 \leq i \leq m$, we extend the run $\langle N, E, L \rangle$ in the following manner:
        \begin{itemize}
            \item If $\FS \models \NFAToPostiveQuery(\FF{NFA}(\DLAut, p), q_i, c)$, there is an accepting run on~$\DLAut$ over~$\FS$ from $\langle q_i, c \rangle$ as argued in the base case.
            We attach this directed tree to the node $r_i$.
            \item Otherwise, there is a constant $t$ and some states $p', p_1,\ldots, p_n \in Q$ such that $q \rightarrow_\NFA^* p'$, $\Transition{p'}{\EmptyTransition}{\{p_1,\ldots p_n\}}$, $q \succ_\DLAut p_1$, \ldots, $q \succ_\DLAut p_n$, and \FS satisfies the following PQ:
            $$\RE_{q_i \triangleright p'}^\NFA(c,t) \wedge \bigwedge\nolimits_{1 \leq i \leq n} \StratDLAToPostiveQuery(\DLAut, p_i, t) $$
            Since \FS satisfies $\RE_{q_i \triangleright p'}^\NFA(c,t)$, we have from \Cref{lemma-nfa-transition} that there is a run whose root node is labeled $\langle q_i, c \rangle$ and whose unique leaf is labeled $\langle p', t \rangle$ : we attach this directed tree to $r_i$, and call the fresh leaf node $r_i'$.
            
                    Since $\FS$ satisfies $\Transition{p'}{\EmptyTransition}{\{p_1, \ldots, p_n\}}$, the edge $\langle p', t \rangle \rightarrow \{ \langle p_j, t \rangle \mid j \leq n \}$ is valid: we attach it to $r_i'$, and label the fresh leaves nodes $r'_{i, j}$. 
                    
                   Since $\FS$ satisfies $\StratDLAToPostiveQuery(\DLAut, p_1, t)$,\ldots, $\StratDLAToPostiveQuery(\DLAut, p_n, t)$, by induction hypothesis, there are accepting runs $\langle N'_j, E'_j, L'_j \rangle$ from $\langle p_1, t \rangle$ , \ldots, $\langle p_n, t \rangle$.  We attach each $\langle N'_j, E'_j, L'_j \rangle$ to $r'_{i,j}$.

        \end{itemize}
    After the modifications, $\langle N, E, L \rangle$ has only accepting leaves: we have built an accepting run for $\DLAut$ over $\FS$ from $\langle p, c \rangle$. We conclude that there is an accepting run of $\DLAut$ over $\FS$ from $\langle p, c \rangle$,  and the induction step holds. Therefore, if some $\FS$ satisfies the PQ $\PQ(\DLAut)$, we have an accepting run of $\DLAut$ from $ \langle \InitialState, s \rangle$ over $\FS$.
    \end{proof}

    \begin{lemma}
    \label{lemma:dl-aut-in-strat-query}
        A DL automaton $\DLAut$ is contained in $\StratToPQ(\DLAut)$.
    \end{lemma}

    \begin{proof}
        We fix a DL automaton $\DLAut= \langle s, Q, \InitialState, \delta \rangle$  and a fact set $\FS$. To prove the lemma, we assume that there is an accepting run of $\DLAut$ from $\langle \InitialState, s \rangle $ over~$\FS$, and show that this implies that $\FS$ satisfies $\StratToPQ(\DLAut, \InitialState,s)$.
         More precisely, we show the following claim by induction on $k \geq 1$: for any state $p$ of $\DLAut$ in stratum $k$ and any constant $c$, if $\DLAut$ has an accepting run from some $\langle p, c \rangle$ we have $\FS \models \StratToPQ (\DLAut, p, c )$.  Applying this claim to $\InitialState$ and $s$, it follows that if there is an accepting run for $\FS$ from $\langle \InitialState, s \rangle$, then $\FS$ satisfies $\StratToPQ(\DLAut, \InitialState, s)$.

        For the base case, consider a state $p$ of stratum 1 and an accepting run $\langle N, E, L \rangle$ whose root is labeled $\langle p,c  \rangle$ and whose leaves are labeled with $ \{ \langle \AcceptingState, w_1 \rangle , \ldots, \langle \AcceptingState, w_n \rangle \}$ for some $w_1, \ldots, w_n$ in $\DLAut$.
        We define the \emph{epsilon-section} of $\langle N, E, L \rangle$, noted $\langle N, E, L \rangle_\EmptyTransition$, as the (maximal) subset run starting from the root using only edges matched with $\EmptyTransition$-transitions.  By the definition of $\delta_\EmptyTransition^*$, there is a transition $\Transition{p}{\EmptyTransition}{Q'} \in \delta_\EmptyTransition^*$ such that the leaves of $\langle N, E, L \rangle_\EmptyTransition$ are labeled with the states of $Q'$. Since $p$ is of stratum 1, there can be only binary edges below these nodes: if that was not the case, some states of leaves of the run would be of a lower strata than the state of the root. Moreover, for every $q \in Q'$, there is an accepting run from $\langle q, c \rangle$ in $\NFA= \FF{NFA}(\DLAut, p)$ over $\FS$. By \Cref{lemma-transition-crpq}, this means that $\NFAToPostiveQuery(\NFA, q, c)$ accepts. Therefore, the PQ $ \bigvee_{\Transition{p}{\EmptyTransition}{Q'}} \bigwedge_{q \in Q'} \NFAToPostiveQuery(\NFA, q, c) = \StratToPQ(\DLAut, \InitialState, c)$ accepts and the property is verified. 

        For the induction step, if $p$ is of strata $k \geq 2$, we consider an accepting run for $\langle p, c  \rangle$ on $\DLAut$. Similarly to the base case, we consider the epsilon-section of the run: if there are only binary edges outside of it, then we are in the same situation as in the base case.
        If not, we consider the leaves $r_1, \ldots, r_m$ of the epsilon-section of the run, labeled $\{ \langle q_i, c \rangle \mid i \leq m\}$. For each $i$, we are in one of two cases:
        \begin{itemize}
    \item If $r_i$ has only binary edges descending from it, then the part of the tree rooted at $r_i$ is an accepting run for $\FF{NFA}(\DLAut, q)$. In that case, $\NFAToPostiveQuery(\FF{NFA}(\DLAut, p), q_i, c)$ is satisfied.
    \item If that is not the case, consider the first multi-ary edge descending from $r_i$, called  $r'_i \to \{ r'_{i,1}, \ldots, r'_{i, n} \}$ with $r'_i$ labeled $\langle p', t \rangle$ and each $r'_{i,j}$ labeled $\langle p_j, t \rangle$.
        \begin{itemize}
            \item By definition, $\langle q, c \rangle$ can transition to $\langle p', t \rangle$ in $\NFA= \FF{NFA}(\DLAut, q)$: therefore, by \Cref{lemma-nfa-transition} we have $\exists w .  \RE_{q \triangleright p'}^\NFA(c, t)$. 
            \item Moreover, in the automaton, we have that $\Transition{p'}{\EmptyTransition}{\{p_1,\ldots, p_n\}}$, and since $p_1$ ,  \ldots, and $p_n$ follow a multi-ary transition outside the epsilon-section of the run $q_i \succ_\DLAut p_1$,  \ldots,  and $q_i \succ_\DLAut p_n$.
            \item Since the run is accepting, the sub-runs starting at $r'_{i,1}$, \ldots,  $r'_{i,n}$ are valid runs with only accepting leaves: they are also accepting for $\DLAut$ over $\FS$. As $p_1,\ldots, p_n$ have strata at most $k$, by induction hypothesis, $ \bigwedge_{1 \leq j \leq n} \FF{PQ}(\DLAut, p_j, t)$ is satisfied by $\FS$.
        \end{itemize}
        \end{itemize}
It follows that 
\begin{align*}
& \bigwedge\nolimits_{i \leq m} \Big(\NFAToPostiveQuery(\NFA, q_i, t)~\vee \\
&\quad \bigvee_{\substack{\forall p', p_1, \ldots, p_n \in Q \textit{ such that } \TRTransRel{\NFA}{q_i}{p'}, \\ \Transition{p'}{\EmptyTransition}{\{p_1,\ldots, p_n\}},~\InHigherStrata{\DLAut}{q_i}{p_1}, \ldots, \text{ and } \InHigherStrata{\DLAut}{q_i}{p_n}}} \big(\exists w. \RegularExpressionNFAStateToState{q_i}{p'}{\NFA}(t, w) \wedge \bigwedge_{1 \leq j \leq n} \StratDLAToPostiveQuery(\DLAut, p_j, w)  \big) \Big)
\end{align*}
        is satisfied for some $Q$, and therefore that $\StratToPQ(\DLAut, p, c)$ is.

        Therefore, we have that for any state $p$ and constant $c$, if there is an accepting run of $\DLAut$ from $\langle p, c \rangle$ over $\FS$, then $\FS$ satisfies $\StratToPQ(\DLAut, p, c )$.
\end{proof}

\section{Claim of \Cref{section:related-work}}

We provide more details regarding our claim that positive fragments of DL-Lite$_{core}$ and DL-Lite$_{\mathcal{R}}$ sets of rules always have stratified associate DL automata.
We assume familiarity with DL-Lite and refer the reader to \cite{calvaneseetal:dllite}; we discuss positive DL-Lite$_{\mathcal{R}}$ since it subsumes DL-Lite$_{core}$.

Let us first highlight that, to be seen as $\HornALCHI$ rules in our normal form, positive DL-Lite$_{\mathcal{R}}$ rules requires the introduction of a placeholder class $\FP{C_\top}$ that replaces $\top$ in every occurrence of a DL-Lite concept of form $\exists \FP{R}.\top$ (or, as is common, simply $\exists \FP{R}$), where $\FP{R}$ is possibly an inverse role.
We add the rule $\rightarrow \FP{C_\top}(x)$ (axiom of form \eqref{rule:top}) to witness this behavior.

Let us now examine the normalization exposed in \Cref{section:horn-alc-to-normalised-horn-alc}.
Getting rid of axioms of form \eqref{rule:subrole} and \eqref{rule:inverse} does not introduce any axiom of form $\eqref{rule:conjunction}$ with $n \geq 2$, as desired.
In \Cref{section:proofs-horn-alc-to-normalised-horn-alc}, we already established that inference rules Rules~{$\mathbf{R}_\leq$}, $\mathbf{R}^-_\leq$, $\mathbf{R}^r_\sqsubseteq$  and $\mathbf{R}_\bot$ never apply. 
Therefore, the only inference rule in \cite[Table~2]{normal-form-horn-shiq} that may introduce axioms of form $\eqref{rule:conjunction}$ with $n \geq 2$ is the rule $\mathbf{R}_\forall$.
To apply, this rule requires an axiom with shape $\FP{A} \sqsubseteq \forall \FP{R}.\FP{B}$.
However, none of the other rules may produce such axiom, therefore it must come directly from our DL-Lite$_\mathcal{R}$ rules.
This implies that $\FP{A} = \FP{C_\top}$.
The body of the inferred rule thus, syntactically, takes the form \eqref{rule:conjunction} with $n \geq 2$.
However, due to our rule $\rightarrow \FP{C_\top}(x)$, it is equivalent to the exact same inferred rule, but in which every occurrence of $\FP{C_\top}$ has been removed from its body, which restores $n \leq 1$.

\end{document}